\documentclass{article}

\usepackage[left=2cm,right=2cm]{geometry}

\usepackage{xspace}
\usepackage{amsmath}
\usepackage{amssymb}
\usepackage{amsthm}
\usepackage{mathpartir}
\usepackage{mathtools}
\usepackage{stmaryrd}
\usepackage{enumitem}
\usepackage{url}
\usepackage{xcolor}
\usepackage{listings}
\usepackage[colorlinks,bookmarks]{hyperref}
\usepackage[capitalise]{cleveref}
\usepackage{tikz}

\definecolor{green}{RGB}{0,130,0}
\definecolor{lightgrey}{RGB}{240,240,240}

\lstdefinelanguage{Dedukti}
{
  inputencoding=utf8,
  extendedchars=true,
  numbers=none,
  numberstyle={},
  tabsize=2,
  basicstyle={\ttfamily\upshape\mdseries},
  backgroundcolor=\color{lightgrey},
  keywords={def,thm,injective},
  sensitive=true,
  keywordstyle=\color{blue},
  morecomment=[s]{(;}{;)},
  commentstyle={\itshape\color{red}},
  string=[b]{"},
  stringstyle=\color{orange},
  showstringspaces=false,
}
\newtheorem{definition}{Definition}
\newtheorem{lemma}{Lemma}

\newtheorem{theorem}{Theorem}
\newtheorem{example}{Example}

\def\ra{\rightarrow}
\def\lra{\hookrightarrow}
\def\Type{\mathsf{Type}}
\def\Kind{\mathsf{Kind}}

\newcommand{\wfr}[3]{#1\vdash #2\lra #3}

\def\dom{\mathsf{dom}}

\def\T{\mathbb{T}}
\def\S{\mathbb{S}}
\def\R{\mathcal{R}}

\newcommand{\rrc}[3]{[#1]\,#2\lra #3} 

\def\mor{\mu} 
\def\lr{\rho} 

\newcommand\dk{Dedukti\xspace}
\newcommand{\lpr}{$\lambda\Pi\slash\R$\xspace}
\newcommand{\lpcm}{$\lambda\Pi$-calculus modulo rewriting\xspace}

\def\imp{\mathbin{\Rightarrow}}
\def\conj{\mathbin{\wedge}}

\def\fa{{\forall}}

\def\Set{{\it Set}}
\def\El{{\it El}}
\def\Prop{{\it Prop}}
\def\Prf{{\it Prf}}
\def\arr{\mathbin{\rightsquigarrow}}
\def\o{o}
\def\impd{\mathbin{\Rightarrow'}}

\newcommand{\oft}[2]{#1 ~\#~ #2}
\def\tm{tm}
\def\lam{lam}

\def\unit{unit}

\def\nat{\mathsf{nat}}
\def\int{\mathsf{int}}
\def\succ{\mathsf{succ}}
\def\pred{\mathsf{pred}}
\def\rec{\mathsf{rec}}
\def\ax{\mathsf{ax}}

\def\mult{\times}
\def\inv{\mathsf{inv}}
\def\i{\iota}
\def\refl{\mathsf{refl}}
\def\leib{\mathsf{leib}}

\def\pair{\mathsf{pair}}
\def\mkpair{\mathsf{mk\_pair}}
\def\fst{\mathsf{fst}}
\def\snd{\mathsf{snd}}

\def\impi{\mathsf{imp_i}}
\def\impe{\mathsf{imp_e}}
\def\andel{\mathsf{and_{e\ell}}}

\def\alli{\mathsf{all_i}}
\def\alle{\mathsf{all_e}}

\def\lst{\mathsf{list}}
\def\nil{\mathsf{nil}}
\def\cons{\mathsf{cons}}
\def\hd{\mathsf{hd}}
\def\tl{\mathsf{tl}}
\def\concat{\mathsf{concat}}
\def\List{\ident{List}}
\def\tree{\mathsf{tree}}
\def\leaf{\mathsf{leaf}}
\def\node{\mathsf{node}}
\def\cright{\mathsf{right}}
\def\cleft{\mathsf{left}}
\def\root{\mathsf{root}}
\def\compo{\mathsf{compo}}
\def\Tree{\ident{Tree}}

\newcommand{\UFOL}{\ident{UFOL}}
\newcommand{\SFOL}{\ident{SFOL}}
\newcommand{\HFOL}{\ident{HFOL}}
\newcommand{\HS}{\ident{hs}}
\newcommand{\SU}{\ident{su}}

\newcommand{\ident}[1]{\ensuremath{\mathtt{#1}}\xspace}

\newcommand{\modif}[1]{#1} 
\newcommand{\reviewernote}[1]{} 

\newtheorem{remark}{Remark}

\begin{document}

\title{Formalizing Representation Theorems for a Logical Framework
with Rewriting}
\date{}
\author{Thomas Traversi\'e\thanks{Universit\'e Paris-Saclay, CentraleSup\'elec, MICS, France and Universit\'e Paris-Saclay,  Inria, CNRS, ENS Paris-Saclay, LMF, France, \texttt{thomas.traversie@centralesupelec.fr}} \and Florian Rabe\thanks{University Erlangen-Nuremberg, Germany, \texttt{florian.rabe@fau.de}}}

\maketitle

\begin{abstract}
Representation theorems for formal systems often take the form of an inductive translation that satisfies certain invariants, which are proved inductively.
Theory morphisms and logical relations are common patterns of such inductive constructions.
They allow representing the translation and the proofs of the invariants as a set of translation rules, corresponding to the cases of the inductions.
Importantly, establishing the invariants is reduced to checking a finite set of, typically decidable, statements.
Therefore, in a framework supporting theory morphisms and logical relations, translations that fit one of these patterns become much easier to formalize and to verify.

The $\lambda\Pi$-calculus modulo rewriting is a logical framework designed for representing and translating between formal systems that has previously not systematically supported such patterns.
In this paper, we extend it with theory morphisms and logical relations.
We apply these to define and verify invariants for a number of translations between formal systems.
In doing so, we identify some best practices that enable us to obtain elegant novel formalizations of some challenging translations, in particular sort-erasure translations from sorted to unsorted languages.
\end{abstract}

\section{Introduction}

\paragraph{Motivation and Related Work}
Logical frameworks are meta-languages for formalizing deductive systems. 
The idea originated in \textsc{Automath}~\cite{automath} and was refined, e.g., by the Isabelle system~\cite{isabelle} based on higher-order logic and the Edinburgh Logical Framework~\cite{LF} (LF) based on the dependently-typed $\lambda$-calculus.
A variety of LF-based practical logical frameworks have been developed, including extensions with logic programming in Twelf~\cite{twelf}, abstraction over contexts in Beluga~\cite{beluga}, monadic side conditions in $\text{LLF}_{\mathcal{P}}$~\cite{llfp}, and user-definable features in MMT~\cite{rabe:recon:17}.
Many different logics can be encoded in LF-based frameworks including higher-order logics \cite{LF}, type systems~\cite{theoryU}, modal logics~\cite{lfmodal}, foundations of mathematics \cite{IR:foundations:10}, model theory \cite{HR:folsound:10}, or the calculus of constructions~\cite{theoryU}.

The \lpcm~\cite{lambdapi} (\lpr) extends LF with user-defined rewrite rules, both at the term and the type level.
All terms and types are considered modulo the congruence relation induced by the usual $\beta$-reduction rule and by the user-defined rewrite rules.
\lpr was implemented in the \dk proof language~\cite{expressing,saillard}, designed for exchanging proofs between systems. 
For instance, it was used to translate~\cite{thire} the Matita arithmetic library to several systems including Rocq and PVS, and to export~\cite{holdk} the HOL Light standard library to Rocq.

\modif{The type-level rewriting includes in particular the ability to rewrite an atomic type into a function type.
For example, the type $\Prf ~(p \imp q)$ of proofs of implications can be rewritten into---and thus become equal to---the function type $\Prf ~p \ra \Prf ~q$.
Compared to typical encodings of natural deduction without rewriting, this means that the implication introduction and elimination rules, which normally map back and forth between those two types, can be dropped entirely.
For example, the modus ponens step $\impe ~p ~q ~H_{pq} ~H_p$, with $\impe:\Pi p,q : \Prop. ~\Prf ~(p \imp q) \ra \Prf ~p \ra \Prf ~q$, becomes simply $H_{pq} ~H_p$, thus erasing the explicit application of $\impe$ and its arguments $p$ and $q$.
This simplifies the work of the formalizer and has a major impact on scalability when representing large libraries of theorems, such as those exported from proof assistants, because it leads to much smaller proof terms.
This effect has been leveraged heavily and critically in the existing encodings in \dk.

Additionally, the need to obtain fast implementations of rewriting 
has led the \lpr community to adopt the design choice of an untyped conversion relation, i.e., both the usual $\beta\eta$-equality and rewriting operate on any terms, not just the well-typed ones.

This raises the question whether meta-theorems about LF carry over, or if these generalizations subtly interfere with established intuitions and results.
This is important because reasoning about the meta-theory of deductive systems, such as establishing type or truth preservation of translations, has been a major application of logical frameworks.
In this paper we answer this question positively for two such meta-theorems as described below.
}

Representation theorems often take the form of $\forall\exists$ meta-statements, e.g., expressing that for all terms $t:A$ of language $\S$, there is a term $\mor(t):\mor(A)$ of language $\T$ (where $\S$ and $\T$ are independently formalized in the framework).
There are two approaches to formalizing such representation functions $\mor$.
Firstly, $\mor$ can be implemented in a programming language that treats the terms of $\S$ and $\T$ as data.
These programs are logic programs in~\cite{twelf} or functional programs in~\cite{beluga,delphin}.
The framework then has to verify the meta-theorem by proving the correctness and termination of the program.
Among early big case studies are verifications of cut-elimination~\cite{lfcut} and logic translations~\cite{hol_nuprl2}.

Secondly, the framework can provide explicit support for certain restricted classes of representation functions for which correctness and termination are guaranteed.
These are usually significantly easier for the user to define, their formalization is more elegant, and their verification is decidable and easy to implement.
When applicable, they are usually the superior formalization.
However, their expressivity is limited, and intended applications often hit these limitations.
\emph{Theory morphisms} (also called signature morphisms) were introduced for LF in~\cite{lfencodings}.
Here the function $\mor$ is induced homomorphically on all $\S$-terms from manually supplied translations of all constants of $\S$.
The function $\mor$ is guaranteed to be total, to preserve all judgments, and to be compositional, i.e., commute with substitutions.
Theory morphisms were added to Isabelle in~\cite{isabelle_locales}, to Twelf in~\cite{RS:twelfmod:09}, and to MMT in~\cite{RK:mmt:10}.
They also allow building a module system in the style of \cite{asl}, and  they were used to build a major library of modular logics and representation theorems~\cite{CHKMR:latinabs:11,rabe:howto:14}.

To extend the expressivity of theory morphisms while retaining their simplicity, \cite{logical_relations} introduced \emph{logical relations} for LF.
Here the function $\mor$ is coupled with a second function $\lr$ that establishes additional invariants about $\mor$.
For example, if $\mor$ is a type-erasure translation, then $\lr$ can be used to state and prove a type-preservation invariant.
Parametricity translations~\cite{parametricity_conf,parametricity_journal} closely resemble logical relations for pure type systems.
They were developed in~\cite{keller_lasson} for the calculus of inductive constructions, and \cite{trocq} builds a parametricity-based Rocq plugin for automated proof transfer.

All of these developments were done in the absence of rewriting.
In fact, other than the Maude tool~\cite{maude}, which is based on membership equational logic with rewriting, we are not aware of any tool supporting theory morphisms or related concepts on top of a rewrite system, and none that do so for a dependently-typed $\lambda$-calculus.
Felicissimo~\cite{felicissimo_encodings} effectively defined theory morphisms in \lpr to establish the soundness of an encoding of functional and explicitly-typed pure type systems.
Similarly, \cite{deduktinterp} encoded individual interpretations in which the morphism and the relation are mutually recursive.
However, both lacked a general definition of the concept and a general meta-theorem establishing their properties once and for all.
That is critical to fully leverage morphisms in practice, because it allows shifting most of the work to the framework and leaving only a small amount of work to the formalizer of an individual translation.

\paragraph{Contribution}
Our work follows the morphism-based approach mentioned above in the context of the \lpcm. 
\modif{Our contribution is threefold.}

Firstly, we generalize theory morphisms and logical relations from LF to \lpr, stating all definitions and theorems in full generality.
\modif{Maybe surprisingly, even though the initial definition of \lpr is more complex than that of LF, the proofs down the road become significantly simpler.
The resulting formalism subsumes the translation templates of~\cite{logical_relations} and simplifies the special interpretations of~\cite{deduktinterp}.}

\modif{Secondly, we apply this formalism to define translations that have previously been out of reach for morphism-like methods.
In particular, we formalize a novel translation from hard-sorted to soft-sorted to unsorted logic.
We show that a specific combination of subtle design choices allows representing these translations concisely and then obtaining their soundness from the framework for free.
To our knowledge, this is the first time that such translations are defined in a fully declarative way.
Notably, some of these design choices involve encoding artefacts that have previously, i.e., in the absence of rewriting, proved to be very cumbersome to use at scale.
Here, rewriting enables us to use clever encodings when defining the translations and then eliminate the artefacts when using the translations in practice. 
}

Thirdly, we have implemented the \texttt{TranslationTemplates} tool, which realizes theory morphisms and unary logical relations for \dk.
We have used \texttt{TranslationTemplates} to formalize all examples shown in the sequel.
The implementation and examples are available at
\begin{center}
\url{https://github.com/Deducteam/TranslationTemplates}.
\end{center}

\paragraph{Overview} 
In \Cref{sec_lpcm}, we recap the syntax and typing rules of \lpr. 
Then we define theory morphisms and logical relations for \lpr in \Cref{sec_mor} and \Cref{sec_lr}.
In \Cref{sec_sort} and \Cref{sec_nat_int}, we apply our framework to challenging sort-erasure translations and datatype embeddings.
We describe our implementation in \Cref{sec_implem}.


\section{The \texorpdfstring{$\lambda\Pi$}{lambdaPi}-Calculus Modulo Rewriting}
\label{sec_lpcm}
The Edinburgh Logical Framework, also known as LF or $\lambda\Pi$-calculus, corresponds to simply typed $\lambda$-calculus extended with dependent types. 
\lpr is an extension of LF with user-defined rewrite rules. 

The terms of \lpr are divided into three levels: objects (denoted by $M$ and $N$), types (denoted by $A$ and $B$), and kinds (denoted by $K$).
The syntax of \lpr is given by the following grammars:
\begin{align*}
&\text{\textit{Objects}} &M, N &\Coloneqq c ~|~ x ~|~ \lambda x : A. ~M ~|~ M ~N \\
&\text{\textit{Types}} &A, B &\Coloneqq a ~|~ \Pi x : A. ~B ~|~ \lambda x : A. ~B ~|~ A ~M \\
&\text{\textit{Kinds}} &K &\Coloneqq \Type ~|~ \Pi x : A. ~K \\
&\text{\textit{Terms}} &t, u, v, T &\Coloneqq M ~|~ A ~|~ K ~|~ \Kind
\end{align*}
where $c$ is an object constant, $a$ is a type constant, and $x$ is a variable. 
Dependent products $\Pi x : A. ~B$ (respectively $\Pi x : A. ~K$) are simply written $A \ra B$ (respectively $A \ra K$) when $x$ does not occur in $B$ (respectively $K$). 

Substitutions $\theta$ associate terms to variables. 
We write $t\theta$ for the result of the capture-avoiding substitution of the term $t$ with respect to $\theta$.
Contexts $\Gamma$ are used to specify the type of the free variables, and
theories $\T$ are used to declare the constants and rewrite rules considered by the users. 
Substitutions, contexts and theories are finite sequences, and are written $\varnothing$ when empty. 

Rewrite rules are pairs $M \lra N$ (respectively $A \lra B$).
\modif{The left- and right-hand side of a rewrite rule may contain free variables.
In \Cref{def:wfr}, we will restrict the rules in such a way that all variable types can be inferred.
Therefore, we omit the context that binds these variables from the notation.}
\begin{align*}
&\text{\textit{\modif{Substitutions}}} &\theta &\Coloneqq \varnothing ~|~ \theta, x \leftarrow N \\
&\text{\textit{Contexts}} &\Gamma &\Coloneqq \varnothing ~|~ \Gamma, x : A \\
&\text{\textit{Theories}} &\T &\Coloneqq \varnothing ~|~ \T,\; c : A ~|~ \T,\; a : K ~|~ \T,\; M \lra N ~|~ \T,\; A \lra B
\end{align*}
We write $\T \vdash$ when the theory $\T$ is well formed, $\vdash_\T \Gamma$ when the context $\Gamma$ is well formed, and $\Gamma \vdash_\T t : T$ when the term $t$ has type $T$ in the context $\Gamma$.
For convenience, $\varnothing \vdash_\T t : T$ is simply written $\vdash_\T t : T$. 
We write $\dom(\Gamma)$ for the domain of a context $\Gamma$, and $\dom(\T)$ for the domain of a theory $\T$.
Well-formed syntax is defined by the typing rules given in \Cref{typrules}, \modif{by \Cref{def:wfr} and \Cref{def:conv}}.

\begin{figure}
\begin{flushleft}
\textbf{Contexts}
\end{flushleft}
\begin{mathpar}
\inferrule*[right={[Empty]}]{ }{\vdash_\T \varnothing}

\inferrule*[right={[Decl] $x \notin \dom(\Gamma)$}]{\vdash_\T \Gamma \\ \Gamma \vdash_\T A : \Type}{\vdash_\T \Gamma, x : A}
\end{mathpar}

\begin{flushleft}
\textbf{Theories}
\end{flushleft}
\begin{mathpar}
\inferrule*[right={[Thy-Empty]}]{ }{\varnothing \vdash}

\inferrule*[right={[Decl-Obj] $c \notin \dom(\T$)}]{\T \vdash \\ \vdash_\T A : \Type}{\T, c : A \vdash}

\inferrule*[right={[Decl-Type] $a \notin \dom(\T$)}]{\T \vdash \\ \vdash_\T K : \Kind}{\T, a : K \vdash}

\inferrule*[right={[Rewr-Obj]}]{\T \vdash \\ \wfr \T M N}{\T, M \lra N \vdash}

\inferrule*[right={[Rewr-Type]}]{\T \vdash \\ \wfr \T A B}{\T, A \lra B \vdash}
\end{mathpar}

\begin{flushleft}
\textbf{Objects}
\end{flushleft}
\begin{mathpar}
\inferrule*[right={[Const-Obj] $c : A \in \T$}]{\vdash_\T \Gamma}{\Gamma \vdash_\T c : A}

\inferrule*[right={[Var] $x : A \in \Gamma$}]{\vdash_\T \Gamma}{\Gamma \vdash_\T x : A}

\inferrule*[right={[Abs-Obj]}]{\Gamma \vdash_\T A : \Type \\ \Gamma, x : A \vdash_\T B : \Type \\ \Gamma, x : A \vdash_\T M : B}{\Gamma \vdash_\T \lambda x : A. ~M : \Pi x : A. ~B}

\inferrule*[right={[App-Obj]}]{\Gamma \vdash_\T M : \Pi x : A. ~B \\ \Gamma \vdash_\T N : A}{\Gamma \vdash_\T M ~N : B[x \leftarrow N]}

\inferrule*[right={[Conv-Type] $A \equiv_{\beta(\eta)\R} B$}]{\Gamma \vdash_\T M : A \\ \Gamma \vdash_\T B : \Type}{\Gamma \vdash_\T M : B}
\end{mathpar}

\begin{flushleft}
\textbf{Types}
\end{flushleft}
\begin{mathpar}
\inferrule*[right={[Const-Type] $a : K \in \T$}]{\vdash_\T \Gamma}{\Gamma \vdash_\T a : K}

\inferrule*[right={[Prod-Type]}]{\Gamma \vdash_\T A : \Type \\ \Gamma, x : A \vdash_\T B : \Type}{\Gamma \vdash_\T \Pi x : A. ~B : \Type}

\inferrule*[right={[Abs-Type]}]{\Gamma \vdash_\T A : \Type \\ \Gamma, x : A \vdash_\T K : \Kind \\ \Gamma, x : A \vdash_\T B : K}{\Gamma \vdash_\T \lambda x : A. ~B : \Pi x : A. ~K}

\inferrule*[right={[App-Type]}]{\Gamma \vdash_\T A : \Pi x : B. ~K \\ \Gamma \vdash_\T M : B}{\Gamma \vdash_\T A ~M : K[x \leftarrow M]}

\inferrule*[right={[Conv-Kind] $K \equiv_{\beta(\eta)\R} K'$}]{\Gamma \vdash_\T A : K \\ \Gamma \vdash_\T K' : \Kind}{\Gamma \vdash_\T A : K'}
\end{mathpar}

\begin{flushleft}
\textbf{Kinds}
\end{flushleft}
\begin{mathpar}
\inferrule*[right={[Sort]}]{\vdash_\T \Gamma}{\Gamma \vdash_\T \Type : \Kind}

\inferrule*[right={[Prod-Kind]}]{\Gamma \vdash_\T A : \Type \\ \Gamma, x : A \vdash_\T K : \Kind}{\Gamma \vdash_\T \Pi x : A. ~K : \Kind}
\end{mathpar}
\caption{Well-formedness rules for the syntax of \lpr}
\label{typrules}
\end{figure}

\modif{
\begin{definition}
\label{def:wfr}
Let $\T$ be a theory and $\ell,r$ be two terms. 
We say that the rewrite rule $\ell\lra r$ is \emph{well-formed}, written $\wfr \T\ell r$, when the following holds:
\begin{enumerate}
\item $\ell$ is algebraic but not simply a variable,
\item all free variables of $r$ are also free in $\ell$,
\item the rule preserves typing, i.e., whenever $\Gamma \vdash_\T \ell\theta : T$ for some substitution $\theta$, then $\Gamma \vdash_\T r\theta : T$.
\end{enumerate}
Here a term is called \emph{algebraic} if it is a variable or the application of a constant to algebraic arguments.
\end{definition}

The first two conditions in \Cref{def:wfr} are straightforward: they restrict the left-hand side to be a pattern and allow inferring the types of all free variables in a rewrite rule from their occurrences on the left-hand side.
But the third condition may appear surprising: \lpr does not require, as one might expect, that the left- and right-hand side are well-formed and have the same type.
Instead, it only requires that every rewrite rule preserves typing whenever a substitution instance of the left-hand side is well-formed.
The most common situation where it helps to have this generality is when we linearize rewrite rules.

\begin{remark}[Left-Linear Rules]
\label{leftrule}
Consider an intrinsically typed encoding of polymorphic lists, using constants $\cons$ and $\hd$ that take the underlying type as a first argument.
The natural rewrite rule is $\hd~a~(\cons~a~h~t)\lra h$.

To reuse existing confluence criteria on left-linear rules, it is however preferable to use the linearized version $\hd~a~(\cons~a'~h~t)\lra h$, even though its left-hand side is ill-formed~\cite{blanqui_lics01,saillard}.
This rule is still type-preserving in the sense of \Cref{def:wfr}, because substitution instances of the left-hand side can only be well-typed if they substitute the same terms for $a$ and $a'$.
\end{remark}

While this type-preservation condition makes it more difficult to check individual rules, the following lemma from~\cite{saillard} states that the usual intuition is a sufficient criterion.

\begin{lemma}[Typed Rewrite Rules]
\label{lem:wfr}
Let $\ell \lra r$ be a rewrite rule that satisfies the first two conditions of \Cref{def:wfr}.
If $\Gamma$ declares the variables of $\ell$, and we have $\Gamma\vdash \ell: T$ and $\Gamma\vdash r:T$, then $\ell \lra r$ preserves typing.
\end{lemma}
}

\modif{
Before defining the conversion relation, we explain two design choices.

\begin{remark}[Untyped Conversion]
\label{rem:conv}
Type systems, such as LF and \lpr, can be presented with typed conversion (where $\Gamma \vdash M \equiv N : A$ means that $M$ and $N$ are equal \emph{and} well-typed) or untyped conversion (where $M \equiv N$ only means that $M$ and $N$ are equal).

In the presence of rewriting, it is strongly preferable to use the untyped variant \cite{lambdapi,expressing,saillard,theoryU}.
It allows defining and efficiently implementing the conversion relation on the context-free syntax without any dependency on the context or the type system.
To avoid accidentally converting well-typed terms into ill-typed terms, the rules \textsc{Conv-Type} and \textsc{Conv-Kind} check that the converted term is well-typed before using it.

Importantly, this means that proofs of meta-theorems that proceed by induction on derivations can be split into two parts: first establishing the result for conversion and then the result for typing.
This is in contrast to presentations with typed conversion as used in~\cite{logical_relations}, where both results must be proved jointly in a much more complicated mutually recursive induction.
Therefore, the proofs of our main results in \Cref{sec_mor} and \Cref{sec_lr} are simpler than those in~\cite{logical_relations} even though the results are more general.
\end{remark}

\begin{remark}[$\eta$-Conversion and Adequacy]
$\eta$-conversion is necessary in LF-style logical frameworks to obtain adequate encodings of bindings via higher-order abstract syntax.
However, the combination of untyped rewriting and $\eta$ is subtle and must be handled carefully~\cite{blanqui_termination}.
Therefore, \lpr is often presented without $\eta$, and Dedukti provides a flag to turn $\eta$ on or off.

In order to retain compatibility with the existing literature on both \lpr (often without $\eta$) and other logical frameworks (usually with $\eta$), we systematically develop our results for both variants.
In particular, all our meta-theorems hold with and and without $\eta$ unless mentioned otherwise, e.g., in \Cref{ex:muldiviso}.
\end{remark}
}

\modif{
We finally define the conversion relation on (not necessarily well-formed) terms, also called definitional equality.

\begin{definition}[Conversion]
\label{def:conv}
We assume that all terms are once and for all quotiented by $\alpha$-equality.
Let $\T$ be a theory.
The relation $\lra_{\beta\R}$ (respectively $\lra_{\beta\eta\R}$) is the smallest relation, closed by term constructors and substitutions, that is generated by $\beta$-reduction (respectively $\beta$-reduction and $\eta$-expansion) and by the rewrite rules of $\T$. 
The conversion $\equiv_{\beta\R}$ (respectively $\equiv_{\beta\eta\R}$) is the reflexive, symmetric, and transitive closure of $\lra_{\beta\R}$ (respectively $\lra_{\beta\eta\R}$). 

We use the subscript $\beta(\eta)\R$ to indicate that either variant can be used if it globally fixed which variant to use.
\end{definition}

We restrict our attention to theories for which $\lra_{\beta(\eta)\R}$ is confluent and terminating, and thus produces normal forms of terms.
}

\newcommand{\MulGr}{\ident{MulGr}}
\newcommand{\DivGr}{\ident{DivGr}}
\newcommand{\PL}{\ident{PL}}
\newcommand{\PLEq}{\ident{PLeq}}

\begin{example}[\PL]
\reviewernote{\modif{To reduce redundancy, we introduce \PL and \PLEq here that will be used in subsequent examples in Sections 2, 3 and 4.}}
\label{ex_impand}
We define the theory \PL, a fragment of propositional logic with implication and conjunction.
$\Prop$ is the type of propositions and $\Prf$ maps a proposition to the type of its proof.
\begin{flalign*}
&\Prop : \Type & \\
&\Prf : \Prop \ra \Type & \\
&\imp : \Prop \ra \Prop \ra \Prop & \\
&\conj : \Prop \ra \Prop \ra \Prop & 
\end{flalign*}
Every natural deduction inference rule is encoded by an axiom, that is a typed constant.
\begin{flalign*}
&\mathsf{imp_i} : \Pi p,q : \Prop. ~(\Prf ~p \ra \Prf ~q) \ra \Prf ~(p \imp q) & \\
&\mathsf{imp_e} : \Pi p,q : \Prop. ~\Prf ~(p \imp q) \ra (\Prf ~p \ra \Prf ~q) & \\
&\mathsf{and_i} : \Pi p : \Prop. ~\Prf ~p \ra \Pi q : \Prop. ~\Prf ~q \ra \Prf ~(p \conj q) & \\
&\mathsf{and_{e\ell}} : \Pi p,q : \Prop. ~\Prf ~(p \conj q) \ra \Prf ~p & \\
&\mathsf{and_{er}} : \Pi p,q : \Prop. ~\Prf ~(p \conj q) \ra \Prf ~q &
\end{flalign*}
\end{example}

\begin{example}[\PLEq]
The theory \PLEq extends propositional logic \PL with equality. 
$\i$ is the type of individuals. 
We define an equality symbol for elements of type $\i$, along with the reflexivity principle and the Leibniz principle.
\begin{flalign*}
&\i : \Type & \\
&{=} : \i \ra \i \ra \Prop & \\
&\refl : \Pi x : \i. ~\Prf ~(x = x) & \\
&\leib : \Pi x,y : \i. ~\Prf ~(x = y) \ra \Pi P : \i \ra \Prop. ~\Prf ~(P ~x) \ra \Prf ~(P ~y) &
\end{flalign*}
\end{example}

\begin{example}[Groups based on multiplication]
\label{ex_mult}
The theory \MulGr of multiplicative groups extends \PLEq with a multiplication symbol $\mult$, an inverse operation $\inv$, and a neutral element $1$.
\begin{flalign*}
&\mult : \i \ra \i \ra \i & &x \mult 1 \lra x & &x \mult (\inv ~x) \lra 1 & &\inv ~1 \lra 1 \\
&1 : \i & &1 \mult x \lra x & &(\inv ~x) \mult x \lra 1 & &\inv ~(\inv ~x) \lra x \\
&\inv : \i \ra \i & &(x \mult y) \mult z \lra x \mult (y \mult z) & 
\end{flalign*}
\modif{We use rewrite rules to capture the usual associativity, neutrality, and inverseness axioms.
Such axioms therefore derive from the rewrite rules. 
For instance, the axiom of associativity is proved by simply combining $\alli$ and $\refl$.}
We benefit from the computational power of \lpr. 
For example, we obtain the conversion $(\inv ~(\inv ~x)) \mult (\inv ~x \mult y) \equiv_{\beta\R} y$ for free.
\end{example}

\begin{example}[Groups based on division]
\label{ex_div}
The theory \DivGr of division groups extends \PLEq with a division operation $\div$ and a neutral element $1$.
\begin{flalign*}
&\div : \i \ra \i \ra \i & &(x \div y) \div z \lra x \div (y \div (1 \div z)) & &x \div 1 \lra x \\
&1 : \i & &1 \div (1 \div x) \lra x & &x \div x \lra 1 
\end{flalign*}
Using these rewrite rules, we have $((y \div x) \div y) \div (1 \div x) \equiv_{\beta\R} 1$ for free.

\modif{Relative to division, we can define the usual group operation $x \mult y$ as $x \div (1 \div y)$. 
More generally, we will show in \Cref{ex:muldivrel} that \DivGr is isomorphic to \MulGr.}
\end{example}

\begin{remark}[Proof Irrelevance]
When representing proof systems in logical frameworks, it is occasionally important to impose irrelevance conditions on certain type symbols, e.g., on the symbol $\Prf$ from our examples to obtain proof irrelevance for the encoded logic.
Whether or not proof irrelevance can be encoded depends on the specific version of \lpr.
The original version introduced in \cite{lambdapi} used a very general form of rewrite rules with context that allows declaring arbitrary rewrite rules such as
\begin{flalign*}
&\unit:\Type && \star:\unit && \rrc{x:\unit}{x}{\star} && \rrc{p:\Prop,H:\Prf~p}{\Prf~p}{\unit}
\end{flalign*}
This introduces a unit type and rewrites every inhabited proof type into it.
But more recent versions such as the one from~\cite{theoryU} do not allow such rewrite rules in order to simplify confluence analysis and to obtain more efficient implementations.
\modif{They disallow rules whose context cannot be inferred from its left-hand side}, e.g., because it is just a variable (like $x$ above) or because of unused variables (like $H$ above).
\modif{Our \Cref{def:wfr} follows this approach.
However, \Cref{ex:muldivrel} assumes for simplicity that the framework offers some way to encode proof irrelevance.
}
\end{remark}



\section{Theory Morphisms}
\label{sec_mor}
\newcommand{\MDGr}{\ident{MulDivGr}}
\newcommand{\DMGr}{\ident{DivMulGr}}

In this section, we define theory morphisms for \lpr, and we prove the basic Judgment Preservation theorem.
As a running example, we give theory morphisms that translate between \MulGr and \DivGr.

\subsection{Formal Definition}

Theory morphisms from theory $\S$ to theory $\T$ are translations that replace the \textit{constants} of $\S$ by \textit{terms} of $\T$. Such terms are the parameters of the translation and must be provided to perform the translation.

\begin{definition}[Theory morphism]
The mapping $\mor$ defined inductively from a set of parameters $\mor_c$ and $\mor_a$ by
\[
\begin{array}{lllclll}
\mor(x) &= &x &&\mor(\lambda x : A. ~M) &= &\lambda x : \mor(A). ~\mor(M) \\
\mor(c) &= &\mor_c &&\mor(\lambda x : A. ~B) &= &\lambda x : \mor(A). ~\mor(B) \\
\mor(a) &= &\mor_a &&\mor(\Pi x : A. ~B) &= &\Pi x : \mor(A). ~\mor(B) \\
\mor(M ~N) &= &\mor(M) ~\mor(N) &&\mor(\Pi x : A. ~K) &= &\Pi x : \mor(A). ~\mor(K) \\
\mor(A ~M) &= &\mor(A) ~\mor(M) &&\mor(\Kind) &= &\Kind \\
\mor(\Type) &= &\Type && \\
\end{array}
\]
is a theory morphism from theory $\S$ to theory $\T$ when:
\begin{enumerate}
\item for every constant $c : A \in \S$, there exists a term $\mor_c$ such that $\vdash_\T \mor_c : \mor(A)$,
\item for every constant $a : K \in \S$, there exists a term $\mor_a$ such that $\vdash_\T \mor_a : \mor(K)$,
\item for every rewrite rule $\ell \lra r \in \S$, we have $\mor(\ell) \equiv_{\beta(\eta)\R} \mor(r)$,
\end{enumerate} 
where $\mor$ is defined on contexts and substitutions by
\[
\begin{array}{lll}
\mor(\varnothing) &= &\varnothing \\
\mor(\Gamma, x : A) &= &\mor(\Gamma), x : \mor(A)\\
\mor(\theta, x \leftarrow M) &= &\mor(\theta), x \leftarrow \mor(M).
\end{array}
\]
\end{definition}

The first two conditions are the same as in LF: the constants of $\S$ must be mapped to terms of $\T$ that have the correct type. 
When extending theory morphisms from LF to \lpr, we require as a third condition that, for every rewrite rule $\ell \lra r$ of $\S$, we have the conversion $\mor(\ell) \equiv_{\beta\R} \mor(r)$ in $\T$ \modif{(or $\mor(\ell) \equiv_{\beta\eta\R} \mor(r)$ if we use the variant of \lpr with $\eta$).}
Under this condition, we will see that \emph{convertibility} is preserved, i.e., if $t \equiv_{\beta(\eta)\R} u$ in $\S$ then $\mor(t) \equiv_{\beta(\eta)\R} \mor(u)$ in $\T$.
\modif{Intuitively, the three conditions of theory morphisms ensure that the target theory has at least the logical and computational strength as the source theory.}

\begin{remark}
Felicissimo~\cite[see Long version]{felicissimo_encodings} required as a third condition that, for every rewrite rule $\ell \lra r$ of $\S$, we have the rewriting $\mor(\ell) \lra^*_{\beta\R} \mor(r)$ in $\T$ (where $\lra^*_{\beta\R}$ is the reflexive and transitive closure of $\lra_{\beta\R}$). 
Under this condition, \emph{rewritability} is preserved, i.e., if $t \lra^*_{\beta\R} u$ in $\S$ then $\mor(t) \lra^*_{\beta\R} \mor(u)$ in $\T$. 
Felicissimo's condition on rewriting is sufficient to prove our condition on conversion, but it is not necessary. 
For instance, consider $A_1, A_2, A_3$ of type $\Type$ in $\S$, with $A_1 \lra A_3$, and $B_1, B_2, B_3$ of type $\Type$ in $\T$, with $B_1 \lra B_2$ and $B_3 \lra B_2$. 
We set $\mor(A_i) = B_i$ for $i \in \llbracket 1, 3 \rrbracket$. 
We indeed have $\mor(A_1) = B_1 \equiv_{\beta\R} B_3 = \mor(A_3)$ in $\T$, but we do not have $\mor(A_1) \lra^*_{\beta\R} \mor(A_3)$.
For the Judgment Preservation theorem, we only need to preserve convertibility, hence our definition of theory morphisms is more general than Felicissimo's definition.
\end{remark}

\modif{\begin{example}[Identity morphism]
Let $\T$ be a theory. The identity morphism from $\T$ to $\T$ maps each constant to itself, and therefore each term to itself. The conditions of theory morphism are trivially satisfied. 
\end{example}}

\begin{example}[Morphism $\MDGr:\MulGr\to\DivGr$]
We define a morphism \MDGr from \MulGr to \DivGr. All the constants of \PLEq are mapped to themselves.
\begin{flalign*}
&\mor(\mult) = \lambda x,y : \i. ~x \div (1 \div y) & \\
&\mor(1) = 1 & \\
&\mor(\inv) = \lambda x : \i. ~1 \div x & 
\end{flalign*}
For every rewrite rule $\ell \lra r$ of the multiplication group, we can easily show that $\mor(\ell)$ and $\mor(r)$ are convertible using the rewrite rules of \DivGr. For the rewrite rule $x \mult (\inv ~x) \lra 1$, we have $\mor(x \mult (\inv ~x)) \equiv_{\beta\R} x \div (1 \div (1 \div x))$. Since $(1 \div (1 \div x)) \lra x$ and $x \div x \lra 1$, we get $\mor(x \mult (\inv ~x)) \equiv_{\beta\R} 1 \equiv_{\beta\R} \mor(1)$.
\end{example}

\begin{example}[Morphism $\DMGr:\DivGr\to\MulGr$]
We define a morphism \DMGr from \DivGr to \MulGr. All the constants of \PLEq are mapped to themselves.
\begin{flalign*}
&\mor(\div) = \lambda x,y : \i. ~x \mult (\inv ~y) & \\
&\mor(1) = 1 &
\end{flalign*}
For every rewrite rule $\ell \lra r$ of \DivGr, we can easily show that $\mor(\ell)$ and $\mor(r)$ are convertible using the rewrite rules of the multiplication group. For the the rewrite rule $1 \div (1 \div x) \lra x$, we have $\mor(1 \div (1 \div x)) \equiv_{\beta\R} 1 \mult (\inv ~(1 \mult \inv ~x))$. Since $1 \mult x \lra x$ and $\inv ~(\inv ~x) \lra x$, we get $\mor(1 \div (1 \div x)) \equiv_{\beta\R} x \equiv_{\beta\R} \mor(x)$.
\end{example}

\modif{To show that \MDGr and \DMGr are isomorphisms, we have to show that the composition $\MDGr;\DMGr$ and the identity morphism of \MulGr are equal, and accordingly for the opposite composition.
In \Cref{ex:muldiviso} we show that the former condition indeed holds \emph{definitionally}, i.e., up to $\equiv_{\beta\eta\R}$.
Later in \Cref{ex:muldivrel}, we show how a \emph{propositional} equality of morphisms can be proved in a situation where definitional equality is not strong enough.}

\begin{example}[Morphism $\MulGr\to\MulGr$]
\label{ex:muldiviso}
The composition $\MDGr;\DMGr$ is a theory morphism from the multiplication group to itself.
In particular, we get the following parameters.
\begin{flalign*}
&\mor(\mult) = \lambda x,y : \i. ~x \mult \inv ~(1 \mult \inv ~y) & \\
&\mor(1) = 1 & \\
&\mor(\inv) = \lambda x : \i. ~1 \mult \inv ~x &
\end{flalign*}
\modif{Note that, thanks to the rewrite rules of \MulGr, we have $\mor(\mult) \equiv_{\beta\R} \lambda x,y : \i. ~x \mult y$ and $\mor(\inv) \equiv_{\beta\R} \lambda x : \i. ~\inv ~x$.
Therefore, the variant of \lpr with the $\eta$-rule suffices to have $\mor(\mult) \equiv_{\beta\eta\R} \mult$ and $\mor(\inv) \equiv_{\beta\eta\R} \inv$.
In that case, the morphism $\MDGr;\DMGr$ and the identity morphism of \MulGr are definitionally equal.}
\end{example}

\subsection{Judgment Preservation Theorem}

The main property of theory morphisms is that this translation preserves convertibility and judgments. 
Once we have specified the parameters of a theory morphism, we can therefore translate any typing judgment from the source theory to the target theory.
In particular, we can transfer proofs between different theories of \lpr.

\begin{theorem}[Judgment Preservation]
\label{thm_morphism}
Let $\mor$ be a theory morphism from $\S$ to $\T$.
\begin{enumerate}
\item If $\vdash_\S \Gamma$, then $\vdash_\T \mor(\Gamma)$.
\item If $\Gamma \vdash_\S M : A$, then $\mor(\Gamma) \vdash_\T \mor(M) : \mor(A)$.
\item If $\Gamma \vdash_\S A : K$, then $\mor(\Gamma) \vdash_\T \mor(A) : \mor(K)$.
\item If $\Gamma \vdash_\S K : \Kind$, then $\mor(\Gamma) \vdash_\T \mor(K) : \Kind$.
\end{enumerate}
\end{theorem}

\modif{This theorem extends the Judgment Preservation theorem of~\cite{logical_relations} from LF to \lpr, and generalizes the one of~\cite{felicissimo_encodings} as we consider a broader definition of theory morphisms.}
The theorem relies on the substitution lemma, which states that morphism and substitution application commute with each other, and on the conversion lemma, which states that convertibility is preserved by the morphism.

\begin{lemma}[Substitution]
\label{prop_subst_mu}
Let $\mor$ be a theory morphism from $\S$ to $\T$, and $\theta$ be a substitution.
Then, \modif{up to $\alpha$-renaming}, we have:
\begin{enumerate}
\item $\mor(M\theta) = \mor(M)\mor(\theta)$,
\item $\mor(A\theta) = \mor(A)\mor(\theta)$,
\item $\mor(K\theta) = \mor(K)\mor(\theta)$.
\end{enumerate}
\end{lemma}

\begin{proof}
\modif{We show the first item by induction on $M$. 
Suppose that $M$ is a variable $x$. 
If $\theta$ does not substitute $x$, then by definition $\mor(\theta)$ does not substitute $x$ either, so $\mor(x\theta) = \mor(x) = \mor(x)\mor(\theta)$. 
If $\theta$ substitutes $x$ by $N$, then by definition $\mor(\theta)$ substitutes $x$ by $\mor(N)$, so $\mor(x\theta) = \mor(N) = \mor(x)\mor(\theta)$.
The other cases are shown using the induction hypotheses.
We prove the second and third items similarly.}
\end{proof}

\begin{lemma}[Conversion]
\label{prop_conv_mu}
Let $\mor$ be a theory morphism from $\S$ to $\T$.
\begin{enumerate}
\item If $A \equiv_{\beta(\eta)\R} B$ in $\S$, then $\mor(A) \equiv_{\beta(\eta)\R} \mor(B)$ in $\T$.
\item If $K \equiv_{\beta(\eta)\R} K'$ in $\S$, then $\mor(K) \equiv_{\beta(\eta)\R} \mor(K')$ in $\T$.
\end{enumerate}
\end{lemma}

\begin{proof}
The proof proceeds by induction on the formation of $A \equiv_{\beta(\eta)\R} B$ and $K \equiv_{\beta(\eta)\R} K'$.
\begin{itemize}
\item By definition, we have $\mor((\lambda x : A. ~M) ~N) = (\lambda x : \mor(A). ~\mor(M)) ~\mor(N)$, which $\beta$-reduces to $\mor(M)\mor([x \leftarrow N])$. 
By \Cref{prop_subst_mu}, we get $\mor((\lambda x : A. ~M) ~N) \equiv_{\beta\R} \mor(M[x \leftarrow N])$. 
We prove $\mor((\lambda x : A. ~B) ~M) \equiv_{\beta\R} \mor(B[x \leftarrow M])$ similarly.
\item By definition, we have $\mor(\lambda x : A. ~M ~x) = \lambda x : \mor(A). ~\mor(M) ~x$, which $\eta$-reduces to $\mor(M)$. 
\item Let $\ell \lra r \in \S$ and $\theta$ be a substitution. 
Using \Cref{prop_subst_mu}, we have $\mor(\ell\theta) = \mor(\ell)\mor(\theta)$ and $\mor(r\theta) = \mor(r)\mor(\theta)$. 
By definition of theory morphisms, we derive $\mor(\ell\theta) = \mor(\ell)\mor(\theta) \equiv_{\beta(\eta)\R} \mor(r)\mor(\theta) = \mor(r\theta)$.
\item Closure by context, reflexivity, symmetry, and transitivity are immediate. \qedhere
\end{itemize}
\end{proof}

\begin{proof}[Proof of \Cref{thm_morphism}]
The proof proceeds by induction on the typing derivations. 
The cases \textsc{App-Obj} and \textsc{App-Type} rely on \Cref{prop_subst_mu}. 
The cases \textsc{Conv-Type} and \textsc{Conv-Kind} rely on \Cref{prop_conv_mu}. 
\end{proof}

\modif{
\begin{remark}
Theory morphisms capture that derivability in the source theory implies derivability in the target theory.
The reverse is not necessarily true and indeed often fails, e.g., if the target has greater logical and computational strength than the source.
We get back to this in \Cref{sec:conc}.
\end{remark}
}

\subsection{Applications}

Theory morphisms encompass many different translations between logics and between data structures. 
We give here three examples.
The first one is a translation from a theory with axiomatized natural deduction to a theory with computational natural deduction.
\modif{The second one is a translation from propositional logic to $Q_0$ logic.
The third one is a translation from lists to binary trees.}

\subsubsection{From Deduction to Computation}
\label{ex_deduc_comput}
\reviewernote{This example has been reworked to refer to \PL}

In \Cref{ex_impand}, we have encoded the natural deduction rules of the implication and the conjunction via typed constants. 
Alternatively, we can make use of the computational power of \lpr and represent natural deduction rules via rewrite rules~\cite{theoryU}.
\begin{flalign*}
&\Prf ~(p \imp q) \lra \Prf ~p \ra \Prf ~q & \\
&\Prf ~(p \conj q) \lra \Pi r : \Prop. ~(\Prf ~p \ra \Prf ~q \ra \Prf ~r) \ra \Prf ~r &
\end{flalign*}
We can now define a theory morphism $\mor$ from the deductive encoding to the computational encoding of the natural deduction rules, thus showing that the latter is a refinement of the former. 
The constants shared by both encodings are mapped to themselves.
The constants representing natural deduction rules are mapped to theorems proving them by using the rewrite rules.
We map the introduction of implication to 
$$\mor(\impi) = \lambda p,q : \Prop. ~\lambda H : \Prf ~p \ra \Prf ~q. ~H$$
of type $\Pi p,q : \Prop. ~(\Prf ~p \ra \Prf ~q) \ra \Prf ~(p \imp q)$, since $\Prf ~p \ra \Prf ~q$ and $\Prf ~(p \imp q)$ are convertible. 
Similarly, the left-elimination of the conjunction is mapped to
$$\mor(\andel) = \lambda p,q : \Prop. ~\lambda H_{pq} : \Prf ~(p \conj q). ~H_{pq} ~p ~(\lambda H_p : \Prf ~p. ~\lambda H_q : \Prf ~q. ~H_p)$$
which has type $\Pi p,q : \Prop. ~\Prf ~(p \conj q) \ra \Prf ~p$. 
The same idea applies for the remaining rules.


\modif{
\subsubsection{From Propositional Logic to $Q_0$ Logic}
\reviewernote{This new example replaces the example about Kuroda translation. It uses \PL and illustrates the advantages of rewriting.}

Andrews formulated $Q_0$~\cite{andrews_q0}, a version of higher-order logic in which the only primitive symbol is equality. 
All the other connectives and quantifiers are then built upon equality.
$Q_0$ is for instance used in the HOL Light proof assistant.

\begin{example}[$Q_0$ logic]
\label{ex_q0}
To encode $Q_0$ in \lpr, we have to define a polymorphic equality.
However, we cannot quantify on $\Type$. 
We define the type $\Set$ of sorts, and $\El$ that maps sorts to the type of their elements. 
In doing so, we can quantify on sorts and embed them into types.
The sort $\o$ encodes propositions.
The arrow $\arr$ encodes function sorts.
The rewrite rule on $\arr$ states that the type embedding a function sort indeed corresponds to a function type~\cite{theoryU}. 
$\Prf$ maps propositions to the type of their proofs. 
\begin{flalign*}
&\Set : \Type & &\arr : \Set \ra \Set \ra \Set & &\o : \Set \\
&\El : \Set \ra \Type & &\El ~(x \arr y) \lra \El ~x \ra \El ~y & &\Prf : \El ~\o \ra \Type
\end{flalign*}
We define a polymorphic equality that satisfies the Leibniz rule, functional extensionality and propositional extensionality. 
For simplicity, ${=} ~a ~x ~y$ is written $x =_a y$. 
Contrary to the usual presentation of $Q_0$, we take advantage of \lpr and define the Leibniz rule through rewriting.
\begin{flalign*}
&{=} : \Pi a : \Set. ~\El ~a \ra \El ~a \ra \El ~\o & \\
&\Prf ~(x =_a y) \lra \Pi P : \El ~a \ra \Prop. ~\Prf ~(P ~x) \ra \Prf ~(P ~y) & \\
&\mathsf{propext} : \Pi p,q : \Prop. ~(\Prf ~p \ra \Prf ~q) \ra (\Prf ~q \ra \Prf ~p) \ra \Prf ~(p =_{\o} q) & \\
&\mathsf{funext} : \Pi a,b : \Set. ~\Pi f,g : \El ~(a \arr b). ~(\Pi x : \El ~a. ~\Prf ~(f ~x =_b g ~x)) \ra \Prf ~(f =_{a \arr b} g) &
\end{flalign*}
In $\mathsf{funext}$, the applications $f ~x$ and $g ~x$ are well-typed thanks to the rewrite rule on $\arr$.
\end{example}

We define a theory morphism from \PL to $Q_0$. 
The constant $\Prf$ is mapped to itself, and $\Prop$ is mapped to $\El ~\o$.
Each connective is mapped to its encoding in $Q_0$.
\begin{flalign*}
&\mor(\conj) = \lambda p,q : \El ~\o. ~(\lambda f. ~f ~p ~q) =_{(\o \arr \o \arr \o) \arr \o} (\lambda f. ~f ~\top ~\top) & \\
&\mor(\imp) = \lambda p,q : \El ~\o. ~(\mor(\conj) ~p ~q) =_{\o} p &
\end{flalign*}
where $\top$ is an abbreviation for the term $(\lambda x. ~x) =_{\o \arr \o} (\lambda x. ~x)$.
Note that the translation of $\conj$ is well-typed only because of the rewrite rule on $\arr$.
The translations of the natural deduction rules rely on $\mathsf{funext}$, on $\mathsf{propext}$ and on intermediate lemmas on equality. 

In LF, i.e., without rewriting, we would have to define an application constructor $app : \Pi a,b : \Set. ~\El ~(a \arr b) \ra \El ~a \ra \El ~b$ and an abstraction constructor $\lam : \Pi a,b : \Set. ~(\El ~a \ra \El ~b) \ra \El ~(a \arr b)$ instead of the rewrite rule $\El ~(x \arr y) \lra \El ~x \ra \El ~y$, along with an axiom for $\beta$-reduction.
Similarly, the Leibniz rule would have been encoded by an axiom.
Therefore, in LF, the user has to explicitly apply these constructors and axioms.
In \lpr, this work is discharged by the framework through rewriting, resulting in much simpler proof obligations for the user.

}

\modif{
\subsection{From Lists to Binary Trees}
\label{sec_list_tree}
\begin{example}[Lists]
The theory \List defines the data structure of lists indexed by the sort of their elements. 
We use $\Set$ and $\El$ defined in \Cref{ex_q0} to define the polymorphic constructor $\lst$, which maps any sort $a$ to the type of the lists containing elements of sort $a$. 
$\nil$ creates an empty list.
$\cons$ appends an element to a list.
$\hd$ returns the head of a list and $\tl$ returns the tail of a list.
$\concat$ concatenates two lists.
\begin{flalign*}
&\lst : \Set \ra \Set & &\cons : \Pi a : \Set. ~\El ~a \ra \El ~(\lst ~a) \ra \El ~(\lst ~a) & &\hd : \Pi a : \Set. ~\El ~(\lst ~a) \ra \El ~a \\
&\nil : \Pi a : \Set. ~\El ~(\lst ~a) & &\concat : \Pi a : \Set. ~\El ~(\lst ~a) \ra \El ~(\lst ~a) \ra \El ~(\lst ~a) & &\tl: \Pi a : \Set. ~\El ~(\lst ~a) \ra \El ~(\lst ~a) 
\end{flalign*}
The semantic of these symbols is encoded using linearized rewrite rules, as explained in \Cref{leftrule}.
\begin{flalign*}
&\hd ~a ~(\cons ~a' ~x ~l) \lra x & \\
&\tl ~a ~(\cons ~a' ~x ~l) \lra l & \\
&\concat ~a ~(\nil ~a') ~l \lra l & \\
&\concat ~a ~(\cons ~a' ~x ~l_1) ~l_2 \lra \cons ~a ~x ~(\concat ~a ~l_1 ~l_2) &
\end{flalign*}
The left-hand side of these rewrite rules are obviously ill typed, but this is not a problem: the rules preserve typing.
Due to the typing constraints of the head symbols, these rewrite rules can only be used when $a$ and $a'$ are instantiated by the same sort.
\end{example}

\begin{example}[Trees]
The theory \Tree defines the data structure of binary trees indexed by the sort of their elements.
$\leaf$ creates the empty tree.
$\node$ creates a new binary tree, by taking as arguments its root and the two children. 
$\cleft$ returns the left child, $\cright$ returns the right child, and $\root$ returns the root.
\begin{flalign*}
&\tree : \Set \ra \Set & &\cleft : \Pi a : \Set. ~\El ~(\tree ~a) \ra \El ~(\tree ~a) \\
&\leaf : \Pi a : \Set. ~\El ~(\tree ~a) & &\cright : \Pi a : \Set. ~\El ~(\tree ~a) \ra \El ~(\tree ~a) \\
&\node : \Pi a : \Set. ~\El ~a \ra \El ~(\tree ~a) \ra \El ~(\tree ~a) \ra \El ~(\tree ~a) & &\root: \Pi a : \Set. ~\El ~(\tree ~a) \ra \El ~a 
\end{flalign*}
We encode the semantics of $\cleft$, $\cright$ and $\root$ using rewriting. 
Again, we use linearized rules.
\begin{flalign*}
&\cleft ~a ~(\node ~a' ~x ~l ~r) \lra l & \\
&\cright ~a ~(\node ~a' ~x ~l ~r) \lra r & \\
&\root ~a ~(\node ~a' ~x ~l ~r) \lra x & 
\end{flalign*}
We define by induction a composition of binary trees. 
The second tree is inductively composed with the left child of the first tree.
\begin{flalign*}
&\compo : \Pi a : \Set. ~\El ~(\tree ~a) \ra \El ~(\tree ~a) \ra \El ~(\tree ~a) & \\
&\compo ~a ~(\leaf ~a') ~t \lra t & \\
&\compo ~a ~(\node ~a' ~x ~l ~r) ~t \lra \node ~a ~x ~(\compo ~a ~l ~t) ~r & 
\end{flalign*}
\end{example}

We can now define a theory morphism from \List to \Tree.
The idea is to represent lists as binary trees with only left children.
For instance, the list $[x_1, x_2, x_3]$ is mapped to the binary tree
\begin{center}
\begin{tikzpicture}[xscale=0.4,yscale=0.8] 
\node(a) at (3,3)[circle,draw,text centered] {$x_1$};
\node(b) at (2,2)[circle,draw,text centered] {$x_2$};
\node(z) at (4,2) {};
\node(c) at (1,1)[circle,draw,text centered] {$x_3$};
\node(y) at (3,1) {};
\node(x) at (2,0) {};
\node(w) at (0,0) {};

\draw[->] (a) -- (b);
\draw[->] (b) -- (c);
\draw[->] (a) -- (z);
\draw[->] (b) -- (y);
\draw[->] (c) -- (x);
\draw[->] (c) -- (w);
\end{tikzpicture}
\end{center}

Formally, the morphism maps $\Set$ and $\El$ to themselves.
$\lst$, $\nil$ and $\concat$ are mapped to their counterparts in \Tree.
The parameters for $\cons$, $\hd$ and $\tl$ are chosen so that we encode lists inside the left of binary trees.
\begin{flalign*}
&\mor(\lst) = \tree & \\
&\mor(\nil) = \leaf & \\
&\mor(\cons) = \lambda a : \Set. ~\lambda x : \El ~a. ~\lambda l : \El ~(\tree ~a). ~\node ~a ~x ~l ~(\leaf ~a) & \\
&\mor(\hd) = \lambda a : \Set. ~\lambda l : \El ~(\tree ~a). ~\root ~a ~l & \\
&\mor(\tl) = \lambda a : \Set. ~\lambda l : \El ~(\tree ~a). ~\cleft ~a ~l & \\
&\mor(\concat) = \lambda a : \Set. ~\lambda l_1, l_2 : \El ~(\tree ~a). ~\compo ~a ~l_1 ~l_2 &
\end{flalign*}
The conditions on the rewrite rules of \List are satisfied in \Tree.
We only show the most interesting case.
\[
\begin{array}{lll}
\mor(\concat ~a ~(\cons ~a ~x ~l_1) ~l_2) &\equiv_{\beta\R} &\compo ~a ~(\node ~a ~x ~l_1 ~(\leaf ~a)) ~l_2 \\
	&\equiv_{\beta\R} &\node ~a ~x ~(\compo ~a ~l_1 ~l_2) ~(\leaf ~a) \\
	&\equiv_{\beta\R} &\mor(\cons ~a ~x ~(\concat ~a ~l_1 ~l_2))´
\end{array}
\]
Note that the conditions on the rewrite rules only need to be satisfied by the versions of the rules \emph{before} linearization. 
Indeed, we only use instances of the linearized rules where the left-hand side is well typed, i.e., we only use instances of the rules before linearization.

We are therefore able to translate lists into binary trees, where both data structures are encoded using the computational power of \lpr.
}




\section{Logical Relations}
\label{sec_lr}
In this section, we extend logical relations in the sense of~\cite{logical_relations} to \lpr, and we prove the main theorem about them, often called the Basic lemma or Abstraction theorem.
\modif{The technical details of logical relations are much trickier than those of theory morphisms.
But our proof of the Abstraction theorem for \lpr will be simpler than the proof for LF~\cite{logical_relations} despite the added generality.}

\subsection{Formal Definition}

A theory morphism $\mor$ maps the judgment $\Gamma \vdash_{\S} M : A$ to the judgment $\mor(\Gamma) \vdash_{\T} \mor(M) : \mor(A)$.
A logical relation $\lr$ on $\mor$ states and proves an additional invariant satisfied by $\mor$: it maps the judgment $\Gamma \vdash_{\S} M : A$ to the judgment $\lr(\Gamma) \vdash_{\T} \lr(M) : \lr(A) ~\mor(M)$.
Here every type $A$ is mapped to a predicate $\lr(A):\mor(A)\to\Type$ and every term $M:A$ is mapped to a proof $\lr(M)$ that $\mor(M)$ satisfies $\lr(A)$.

The function $\lr$ duplicates every (free or bound) variable so that every fresh variable $x:A$ yields both its translation $x:\mor(A)$ and an assumption $x^*:\lr(A)~x$ that $x$ satisfies the invariant.
Thus, the translation of an abstraction $\lambda x : A. ~M$ is $\lambda x : \mor(A). ~\lambda x^* : \lr(A) ~x. ~\lr(M)$.
Accordingly, the translation of an application $M ~N$ is $\lr(M) ~\mor(N) ~\lr(N)$, i.e., it supplies both the translated argument and the proof of its invariant.

The definition of the logical relation of a function type $\Pi x : A. ~B$ is the well-known condition that functions must preserve the relation: $\lr(\Pi x : A. ~B)$ holds for a function $f:\mor(\Pi x : A. ~B)$ if for every $x:\mor(A)$ satisfying $\lr(A)$, the term $(f ~x)$ satisfies $\lr(B)$.
Thus, $\lr(\Pi x : A. ~B)$ is given by $\lambda f : \mor(\Pi x : A. ~B). ~\Pi x : \mor(A). ~\Pi x^* : \lr(A) ~x. ~\lr(B) ~(f ~x)$.

Following the same idea, we would like to define $\lr(\Pi x : A. ~K) = \lambda f : \mor(\Pi x : A. ~K). ~\Pi x : \mor(A). ~\Pi x^* : \lr(A) ~x. ~\lr(K) ~(f ~x)$.
This works, with a little extra effort, in type theories with higher universes.
However, in \lpr, such a term is ill typed because $\mor(\Pi x : A. ~K)$ is a kind.
To get around this issue, we insert an extra parameter $R$ to the translation: \modif{if $R$ has type $\Pi x : A. ~K$ then} we define $\lr^{R}(\Pi x : A. ~K) = \Pi x : \mor(A). ~\Pi x^* : \lr(A) ~x. ~\lr^{R ~x}(K)$, and \modif{if $R$ has type $\Type$ then we define} $\lr^{R}(\Type) = \mor(R) \ra \Type$.

More generally, we can define $n$-ary logical relations for theory morphisms $\mor_1, \ldots, \mor_n$ from $\S$ to $\T$. 
Such a $n$-ary logical relation $\lr$ maps every type to an $n$-ary predicate and every term $M:A$ to a proof of $\lr(A)~\mor_1(M)~\ldots~\mor_n(M)$.

\paragraph*{Conventions} Let $\mor_1, \ldots, \mor_n$ be $n$ theory morphisms from $\S$ to $\T$. 
Without loss of generality, we consider that each $\mor_i$ maps variables $x$ to $x_i$. 
We use the following notations:
\begin{itemize}
\item We write $\vec{x} : \vec{\mor}(A)$ for the context $x_1 : \mor_1(A), \ldots, x_n : \mor_n(A)$.
\item We write $[\vec{x} \leftarrow \vec{\mor}(M)]$ for the substitution $[x_1 \leftarrow \mor_1(M), \ldots, x_n \leftarrow \mor_n(M)]$.
\item We write $\lambda \vec{x} : \vec{\mor}(A). ~t$ for $\lambda x_1 : \mor_1(A). ~\ldots ~\lambda x_n : \mor_n(A). ~t$.
\item We write $\Pi \vec{x} : \vec{\mor}(A). ~t$ for $\Pi x_1 : \mor_1(A). ~\ldots ~\Pi x_n : \mor_n(A). ~t$. Similarly, we write $\vec{\mor}(A) \ra t$ for $\mor_1(A) \ra \ldots \ra \mor_n(A) \ra t$.
\item Given a list of terms $\vec{M} = M_1, \ldots, M_n$, we write $t ~\vec{M}$ for the application $t ~M_1 ~\ldots ~M_n$.
\end{itemize}

\begin{definition}[Logical relation]
Let $\mor_1, \ldots, \mor_n$ be theory morphisms from $\S$ to $\T$. 
The mapping $\lr$ defined inductively from a set of parameters $\lr_c$ and $\lr_a$ by
\[
\begin{array}{lll}
\lr(x) &= &x^* \\
\lr(c) &= &\lr_c \\
\lr(a) &= &\lr_a \\
\lr(M ~N) &= &\lr(M) ~\vec{\mor}(N) ~\lr(N) \\
\lr(A ~M) &= &\lr(A) ~\vec{\mor}(M) ~\lr(M) \\
\lr(\lambda x : A. ~M) &= &\lambda \vec{x} : \vec{\mor}(A). ~\lambda x^* : \lr(A) ~\vec{x}. ~\lr(M) \\
\lr(\lambda x : A. ~B) &= &\lambda \vec{x} : \vec{\mor}(A). ~\lambda x^* : \lr(A) ~\vec{x}. ~\lr(B) \\
\lr(\Pi x : A. ~B) &= &\lambda \vec{f} : \vec{\mor}(\Pi x : A. ~B). ~\Pi \vec{x} : \vec{\mor(A)}. ~\Pi x^* : \lr(A) ~\vec{x}. ~\lr(B) ~(f_1 ~x_1) ~\ldots ~(f_n ~x_n) \\
\lr^{R}(\Pi x : A. ~K) &= &\Pi \vec{x} : \vec{\mor}(A). ~\Pi x^* : \lr(A) ~\vec{x}. ~\lr^{R ~x}(K) \\
\lr^{R}(\Type) &= &\vec{\mor}(R) \ra \Type \\
\lr(\Kind) &= &\Kind
\end{array}
\]
is a logical relation on $\mu$ when:
\begin{enumerate}
\item for every constant $c : A \in \S$, there exists a term $\lr_c$ such that $\vdash_\T \lr_c : \lr(A) ~\vec{\mor}(c)$,
\item for every constant $a : K \in \S$, there exists a term $\lr_a$ such that $\vdash_\T \lr_a : \lr^{a}(K)$,
\item for every rewrite rule $\ell \lra r \in \S$, we have $\lr(\ell) \equiv_{\beta(\eta)\R} \lr(r)$,
\end{enumerate}
where $\lr$ is defined on contexts and substitutions by
\[
\begin{array}{lll}
\lr(\varnothing) &= &\varnothing \\
\lr(\Gamma, x : A) &= &\lr(\Gamma),\; \vec{x} : \vec{\mor}(A),\; x^* : \lr(A) ~\vec{x} \\
\lr(\theta, x \leftarrow N) &= &\lr(\theta),\; \vec{x} \leftarrow \vec{\mor}(N),\; x^* \leftarrow \lr(N).
\end{array}
\]
\end{definition}

The first two conditions are the same as in LF.
The third condition, specific to \lpr, is necessary so that convertibility is preserved by logical relations.

Note that for every variable $x$ that occurs in a term $t$, the two variables $x_i$ and $x^*$ occur in the translated term $\lr_i(t)$.
Therefore, if $\Gamma$ declares $m$ variables, $\lr(\Gamma)$ declares $m(n+1)$ variables.
\medskip

\modif{
In \Cref{ex:muldiviso}, we discussed how the equality between $\MDGr;\DMGr$ and the identity of \MDGr can be offloaded to the definitional equality of \lpr extended with $\eta$-reduction. 
In general, however, definitional equality---even with $\eta$-reduction---is not sufficient to show the equality between two morphisms, and we have to resort to a propositional equality.
Those situations are one of the applications of logical relations:}

\begin{example}\label{ex:muldivrel}
\modif{In \Cref{ex:muldiviso}, we showed that $\MDGr;\DMGr$ and the identify of \MulGr are definitionally equal in the \lpr meta-logic if we use $\eta$.
But for the sake of example, we now show how to prove propositional equality up to object-logic equality on terms and an object-logic encoded equivalence on propositions.}
To formalize that argument, we capture the invariant as a binary logical relation on $\MDGr;\DMGr$ and on the identity of \MulGr.

We first define a binary logical relation on \PLEq that captures our proof obligations: the two translations of any element of type $\i$ must be equal, and the two translations of any proposition must be equivalent.
\begin{flalign*}
&\lr(\i) = \lambda x_1,x_2 : \i. ~\Prf ~(x_1 = x_2) & \\
&\lr(\Prop) = \lambda p_1,p_2 : \Prop. ~\Prf ~((p_1 \imp p_2) \conj (p_2 \imp p_1)) & 
\end{flalign*}
We also have to define an invariant for proofs.
\modif{Without proof-irrelevance, our logical relation must use a trivially true relation on proof terms.}
That is inessential if we assume proof-irrelevance: then we can simply skip the proof obligations for proofs.
\begin{flalign*}
&\lr(\Prf) = \lambda p_1,p_2 : \Prop. ~\lambda H : \Prf ((p_1 \imp p_2) \conj (p_2 \imp p_1)). ~\unit & 
\end{flalign*}
Defining $\lr$ for the remaining constants is straightforward.

We extend $\lr$ to the constants of \MulGr as follows.
The parameter $\lr(\mult)$ is a proof that if $x_1 = x_2$ and $y_1 = y_2$ then $x_1 \mult y_1 = x_2 \mult y_2$.
\begin{flalign*}
\lr(\mult) = ~&\lambda x_1,x_2 : \i. ~\lambda H_x : \Prf ~(x_1 = x_2). ~\lambda y_1,y_2 : \i. ~\lambda H_y : \Prf ~(y_1 = y_2). & \\ 
&\leib ~x_1 ~x_2 ~H_x ~(\lambda z. ~x_1 \mult y_1 = z \mult y_2) ~(\leib ~y_1 ~y_2 ~H_y ~(\lambda z. ~x_1 \mult y_1 = x_1 \mult z) ~(\refl ~(x_1 \mult y_1))) & 
\end{flalign*}
The parameter $\lr(1)$ is a proof that $1 = 1$, and the parameter $\lr(\inv)$ is a proof that if $x_1 = x_2$ then $\inv ~x_1 = \inv ~x_2$.
\begin{flalign*}
&\lr(1) = \refl ~1 & \\
&\lr(\inv) = \lambda x_1,x_2 : \i. ~\lambda H : \Prf ~(x_1 = x_2). ~\leib ~x_1 ~x_2 ~H ~(\lambda z : \i. ~\inv ~x_1 = \inv ~z) ~(\refl ~(\inv ~x_1)) &
\end{flalign*}
We can easily express the remaining parameters.

To ensure this is a well-formed logical relation, we have to check the condition on the rewrite rules.
For every rewrite rule $\ell \lra r$ of \MulGr, $\ell$ and $r$ have type $\i$.
Therefore $\lr(\ell)$ and $\lr(r)$ are proofs, so that these conditions are trivial under proof irrelevance.
\end{example}

\subsection{Abstraction Theorem}

We have seen that, if $M$ has type $A$, we want $\lr(M)$ to be of type $\lr(A) ~\mu(M)$, where $\lr(A)$ is intuitively an invariant that must be satisfied by $\mu(M)$. 
The Abstraction theorem extends this idea to the three-level hierarchy of LF and \lpr.

\begin{theorem}[Abstraction]
\label{thm_logical_relation}
Let $\lr$ be a logical relation on $\mor_1, \ldots, \mor_n$.
\begin{enumerate}
\item If $\vdash_\S \Gamma$, then $\vdash_\T \lr(\Gamma)$.
\item If $\Gamma \vdash_\S M : A$ and $\Gamma \vdash_\S A : \Type$, then $\lr(\Gamma) \vdash_\T \lr(M) : \lr(A) ~\vec{\mor}(M)$.
\item If $\Gamma \vdash_\S A : K$ and $\Gamma \vdash_\S K : \Kind$, then $\lr(\Gamma) \vdash_\T \lr(A) : \lr^{A}(K)$.
\item If $\ \vdash_\S K : \Kind$ and $\Gamma \vdash_\S A : K$, then $\lr(\Gamma) \vdash_\T \lr^{A}(K) : \Kind$.
\end{enumerate}
\end{theorem}

\modif{Note that, in~\cite{logical_relations}, the preservation of conversion is part of the Abstraction theorem. 
This is because in~\cite{logical_relations}, the conversion is typed, so that both conversion and typing depend on each other and both proofs must be intertwined.
With our untyped conversion, the preservation of conversion can be proved separately and before proving the preservation of typing, making the proofs much simpler.
Our proof needs only a simple substitution lemma and a conversion lemma.}

\begin{lemma}[Substitution]
\label{prop_subst_lr}
Let $\lr$ be a logical relation on $\mor_1, \ldots, \mor_n$, and $\theta$ be a substitution. 
Then, \modif{up to $\alpha$-renaming}, we have:
\begin{enumerate}
\item $\lr(M\theta) = \lr(M)\lr(\theta)$,
\item $\lr(A\theta) = \lr(A)\lr(\theta)$,
\item $\lr^{A\theta}(K\theta) = \lr^{A}(K)\lr(\theta)$.
\end{enumerate}
\end{lemma}

\begin{proof}
By induction on the terms $M$, $A$ and $K$.
\end{proof}

\begin{lemma}[Conversion]
\label{prop_conv_lr}
Let $\lr$ be a logical relation on $\mor_1, \ldots, \mor_n$.
\begin{enumerate}
\item If $A \equiv_{\beta(\eta)\R} B$ in $\S$, then $\lr(A) \equiv_{\beta(\eta)\R} \lr(B)$ in $\T$.
\item If $K \equiv_{\beta(\eta)\R} K'$ in $\S$ then $\lr^{R}(K) \equiv_{\beta(\eta)\R} \lr^{R}(K')$ in $\T$.
\end{enumerate}
\end{lemma}

\begin{proof}
The proof proceeds by induction on the formation of $A \equiv_{\beta(\eta)\R} B$ and $K \equiv_{\beta(\eta)\R} K'$.
\begin{itemize}
\item We have $\lr((\lambda x : A. ~M) ~N) = (\lambda \vec{x} : \vec{\mor}(A). ~\lambda x^* : \lr(A) ~\vec{x}. ~\lr(M)) ~\vec{\mor}(N) ~\lr(N)$, which $\beta$-reduces to $\lr(M)\lr([x \leftarrow N])$. 
By \Cref{prop_subst_lr}, we derive $\lr((\lambda x : A. ~M) ~N) \equiv_{\beta\R} \lr(M[x \leftarrow N])$. 
Similarly, we have $\lr((\lambda x : A. ~B) ~M) \equiv_{\beta\R} \lr(B[x \leftarrow M])$.
\item We have $\lr(\lambda x : A. ~M ~x) = \lambda \vec{x} : \vec{\mor}(A). ~\lambda x^* : \lr(A) ~\vec{x}. ~\lr(M) ~\vec{x} ~x^*$, which $\eta$-reduces to $\lr(M)$. 
\item Let $\ell \lra r \in \S$ and $\theta$ be a substitution. 
Using \Cref{prop_subst_lr}, we have $\lr(\ell\theta) = \lr(\ell)\lr(\theta)$ and $\lr(r\theta) = \lr(r)\lr(\theta)$. 
By definition, we derive $\lr(\ell\theta) = \lr(\ell)\lr(\theta) \equiv_{\beta(\eta)\R} \lr(r)\lr(\theta) = \lr(r\theta)$.
\item Closure by context, reflexivity, symmetry, and transitivity are immediate and relies on \Cref{prop_conv_mu}. \qedhere
\end{itemize}
\end{proof}

\begin{proof}[Proof of \Cref{thm_logical_relation}]
We proceed by induction on the derivations.
\begin{itemize}
\item \underline{\textsc{Empty}}: Since $\lr(\varnothing) = \varnothing$, we derive $\vdash_\T \lr(\varnothing)$ using \textsc{Empty}.

\item \underline{\textsc{Decl}}: 
By induction, we have $\vdash_\T \lr(\Gamma)$ and $\lr(\Gamma) \vdash_\T \lr(A) : \vec{\mor}(A) \ra \Type$. 
Using \Cref{thm_morphism}, we have $\mor_i(\Gamma) \vdash_\T \mor_i(A) : \Type$.
Since $x \notin \Gamma$, we have $x_i \notin \lr(\Gamma)$ and $x^* \notin \lr(\Gamma)$. 
We derive $\vdash_\T \lr(\Gamma), \vec{x} : \vec{\mor}(A), x^* : \lr(A) ~\vec{x}$ using weakening and \textsc{Decl} several times.

\item \underline{\textsc{Sort}}: 
Suppose that $\Gamma \vdash_\S A : \Type$ for some $A$. 
We have $\vdash_\T \lr(\Gamma)$ by induction hypothesis. 
Using \Cref{thm_morphism}, we get $\mor_i(\Gamma) \vdash_\T \mor_i(A) : \Type$. 
Using weakening and \textsc{Prod-Kind} several times, we derive $\lr(\Gamma) \vdash_\T \vec{\mor}(A) \ra \Type : \Kind$.

\item \underline{\textsc{Const-Obj}}: 
We get $\vdash_\T \lr(\Gamma)$ by induction hypothesis. 
We know that $\vdash_\T \lr_c : \lr(A) ~\vec{\mor}(c)$. 
Using weakening, we derive $\lr(\Gamma) \vdash_\T \lr_c : \lr(A) ~\mor_1(c) ~\ldots ~\mor_n(c)$.

\item \underline{\textsc{Const-Type}}: 
We get $\vdash_\T \lr(\Gamma)$ by induction hypothesis. 
We know that $\vdash_\T \lr_a : \lr^{a}(K)$. 
Using weakening, we derive $\lr(\Gamma) \vdash_\T \lr_a : \lr^{a}(K)$.

\item \underline{\textsc{Var}}: 
By induction, we have $\vdash_\T \lr(\Gamma)$. 
Since $x : A \in \Gamma$, we have $x^* : \lr(A) ~\vec{x} \in \lr(\Gamma)$. 
We derive $\lr(\Gamma) \vdash_\T x^* : \lr(A) ~\vec{x}$ using \textsc{Var}.

\item \underline{\textsc{Prod-Type}}: 
By induction, we have 
\[
\begin{array}{ll}
& \lr(\Gamma) \vdash_\T \lr(A) : \vec{\mor}(A) \ra \Type \\
\text{and} & \lr(\Gamma), \vec{x} : \vec{\mor}(A), x^* : \lr(A) ~\vec{x} \vdash_\T \lr(B) : \vec{\mor}(B) \ra \Type.
\end{array} 
\]
Using weakening and then \textsc{App-Obj} and \textsc{App-Type} several times, we get 
\[
\lr(\Gamma), \vec{f} : \vec{\mor}(\Pi x : A. ~B), \vec{x} : \vec{\mor}(A), x^* : \lr(A) ~\vec{x} \vdash_\T \lr(B) ~(f_1 ~x_1) ~\ldots ~(f_n ~x_n) : \Type.
\]
Using \textsc{Prod-Type} and \textsc{Abs-Type} several times, we derive 
\[
\begin{array}{ll}
\lr(\Gamma) \vdash_\T &\lambda \vec{f} : \vec{\mor}(\Pi x : A. ~B). ~\Pi \vec{x} : \vec{\mor}(A). ~\Pi x^* : \lr(A) ~\vec{x}. \\
&\lr(B) ~(f_1 ~x_1) ~\ldots ~(f_n ~x_n) : \vec{\mor}(\Pi x : A. ~B) \ra \Type.
\end{array}
\]

\item \underline{\textsc{Prod-Kind}}: 
Suppose that $\Gamma \vdash_\S R : \Pi x : A. ~K$ for some $R$. 
By induction, we have 
\[
\begin{array}{ll}
& \lr(\Gamma) \vdash_\T \lr(A) : \vec{\mor}(A) \ra \Type \\
\text{and} & \lr(\Gamma), \vec{x} : \vec{\mor}(A), x^* : \lr(A) ~\vec{x} \vdash_\T \lr^{R ~x}(K) : \Kind.
\end{array} 
\] 
Using \textsc{Prod-Kind} several times, we derive $\lr(\Gamma) \vdash_\T \Pi \vec{x} : \vec{\mor}(A). ~\Pi x^* : \lr(A) ~\vec{x}. ~\lr^{R ~x}(K) : \Kind$.

\item \underline{\textsc{Abs-Obj}}: 
By induction, we have 
\[
\begin{array}{ll}
& \lr(\Gamma) \vdash_\T \lr(A) : \vec{\mor}(A) \ra \Type \\
\text{and} & \lr(\Gamma), \vec{x} : \vec{\mor}(A), x^* : \lr(A) ~\vec{x} \vdash_\T \lr(B) : \vec{\mor}(B) \ra \Type \\
\text{and} & \lr(\Gamma), \vec{x} : \vec{\mor}(A), x^* : \lr(A) ~\vec{x} \vdash_\T \lr(M) : \lr(B) ~\vec{\mor}(M).
\end{array}
\]
Using \textsc{Abs-Obj} several times, we derive 
\[
\lr(\Gamma) \vdash_\T \lambda \vec{x} : \vec{\mor}(A). ~\lambda x^* : \lr(A) ~\vec{x}. ~\lr(M) : \Pi \vec{x} : \vec{\mor}(A). ~\Pi x^* : \lr(A) ~\vec{x}. ~\lr(B) ~\vec{\mor}(M).
\]
Using \textsc{Conv-Type}, we get $\lr(\Gamma) \vdash_\T \lr(\lambda x : A. ~M) : \lr(\Pi x : A. ~B) ~\vec{\mor}(\lambda x : A. ~M)$.

\item \underline{\textsc{Abs-Type}}: 
By induction, we have 
\[
\begin{array}{ll}
& \lr(\Gamma) \vdash_\T \lr(A) : \vec{\mor}(A) \ra \Type \\
\text{and} & \lr(\Gamma), \vec{x} : \vec{\mor}(A), x^* : \lr(A) ~\vec{x} \vdash_\T \lr^{B}(K) : \Kind \\
\text{and} & \lr(\Gamma), \vec{x} : \vec{\mor}(A), x^* : \lr(A) ~\vec{x} \vdash_\T \lr(B) : \lr^{B}(K).
\end{array}
\]
We derive $\lr(\Gamma) \vdash_\T \lambda \vec{x} : \vec{\mor}(A). ~\lambda x^* : \lr(A) ~\vec{x}. ~\lr(B) : \Pi \vec{x} : \vec{\mor}(A). ~\Pi x^* : \lr(A) ~\vec{x}. ~\lr^{B}(K)$ using \textsc{Abs-Type} several times. 
Using \textsc{Conv-Kind} and the fact that $\lr^{(\lambda x : A. ~B) ~x}(K) \equiv_{\beta\R} \lr^{B}(K)$, we get $\lr(\Gamma) \vdash_\T \lr(\lambda x : A. ~B) : \lr^{\lambda x : A. ~B}(\Pi x : A. ~K)$.

\item \underline{\textsc{App-Obj}}: 
By induction, we have 
\[
\begin{array}{ll}
& \lr(\Gamma) \vdash_\T \lr(M) : \Pi \vec{x} : \vec{\mor}(A). ~\Pi x^* : \lr(A) ~\vec{x}. ~\lr(B) ~(\mor_1(M) ~x_1) ~\ldots ~(\mor_n(M) ~x_n) \\ \text{and} & \lr(\Gamma) \vdash_\T \lr(N) : \lr(A) ~\vec{\mor}(N).
\end{array}
\]
Using \Cref{thm_morphism} and weakening, we get $\lr(\Gamma) \vdash_\T \mor_i(N) : \mor_i(A)$. 
Using \textsc{App-Obj} several times, we derive 
\[
\lr(\Gamma) \vdash_\T \lr(M) ~\vec{\mor}(N) ~\lr(N) : (\lr(B) ~(\mor_1(M) ~x_1) ~\ldots ~(\mor_n(M) ~x_n))[\vec{x} \leftarrow \vec{\mor}(N), x^* \leftarrow \lr(N)],
\]
that is 
\[
\lr(\Gamma) \vdash_\T \lr(M) ~\vec{\mor}(N) ~\lr(N) : \lr(B)[\vec{x} \leftarrow \vec{\mor}(N), x^* \leftarrow \lr(N)] ~(\mor_1(M) ~\mor_1(N)) ~\ldots ~(\mor_n(M) ~\mor_n(N)).
\]
Using \Cref{prop_subst_lr}, we derive $\lr(\Gamma) \vdash_\T \lr(M ~N) : \lr(B[x \leftarrow N]) ~\vec{\mor}(M ~N)$.

\item \underline{\textsc{App-Type}}: 
By induction, we have $\lr(\Gamma) \vdash_\T \lr(A) : \Pi \vec{x} : \vec{\mor}(B). ~\Pi x^* : \lr(B) ~\vec{x}. ~\lr^{A ~x}(K)$ 
and $\lr(\Gamma) \vdash_\T \lr(M) : \lr(B) ~\vec{\mor}(M)$. 
Using \Cref{thm_morphism} and weakening, we get 
\[
\lr(\Gamma) \vdash_\T \mor_i(M) : \mor_i(B).
\] 
Using \textsc{App-Type} several times, we derive 
\[
\lr(\Gamma) \vdash_\T \lr(A) ~\vec{\mor}(M) ~\lr(M) : \lr^{A ~x}(K)[\vec{x} \leftarrow \vec{\mor}(M), x^* \leftarrow \lr(M)].
\] 
Using \Cref{prop_subst_lr}, we obtain $\lr(\Gamma) \vdash_\T \lr(A ~M) : \lr^{A ~M}(K[x \leftarrow M])$.

\item \underline{\textsc{Conv-Type}}: 
By induction, we have 
\[
\begin{array}{ll}
& \lr(\Gamma) \vdash_\T \lr(M) : \lr(A) ~\vec{\mor}(M) \\
\text{and} & \lr(\Gamma) \vdash_\T \lr(B) : \vec{\mor}(B) \ra \Type.
\end{array}
\]
Since we have $A \equiv_{\beta(\eta)\R} B$ in $\S$, we have $\lr(A) \equiv_{\beta(\eta)\R} \lr(B)$ in $\T$ using \Cref{prop_conv_lr}. 
We derive $\lr(\Gamma) \vdash_\T \lr(M) : \lr(B) ~\vec{\mor}(M)$ using \textsc{Conv-Type}.

\item \underline{\textsc{Conv-Kind}}: 
By induction, we have $\lr(\Gamma) \vdash_\T \lr(A) : \lr^{A}(K)$ and $\lr(\Gamma) \vdash_\T \lr(K') : \Kind$. 
Since $K \equiv_{\beta(\eta)\R} K'$ in $\S$, we have $\lr^{A}(K) \equiv_{\beta(\eta)\R} \lr^{A}(K')$ in $\T$ using \Cref{prop_conv_lr}. 
We derive $\lr(\Gamma) \vdash_\T \lr(A) : \lr^{A}(K')$ using \textsc{Conv-Kind}. \qedhere
\end{itemize}
\end{proof}



\section{Sort-Erasure Translations}
\label{sec_sort}
\reviewernote{This section was heavily revised to emphasize the novelty and technical difficulty of the translations as well as their interplay with rewriting.
We recommend rereading the whole section. But we have marked the most important changes in addition.}

\modif{In this section, we represent sort-erasure translations from hard-sorted to soft-sorted to unsorted logic.
Here ``sort'' is the word that we will use for object-logic types to avoid any confusion with the types of \lpr.
The first translation erases hard sorting and instead captures the sorting relation via a special judgment.
The second translation erases sorts entirely and captures them as unary predicates.
Both of these translations have proved challenging, and to our knowledge, this is the first time that such translations are fully defined and verified.
Our treatment will reveal several subtle critical design choices.

A key insight to formalize these is to adjust the target theories with additional features that allow carrying the sorting properties throughout the translation.
Specifically, we encode dependent pairs to group a term with a proof about it, and we add a dependent variant of implication in order to access such proofs while building propositions.
}

\subsection{Hard-Sorted, Soft-Sorted and Unsorted Logic}

We first formalize unsorted logic \UFOL, soft-sorted logic \SFOL, and hard-sorted logic \HFOL. 

\begin{example}[Unsorted Logic]
In unsorted logic \UFOL, all terms have the generic type $\tm$.
\begin{flalign*}
&\tm : \Type & 
&\Prop : \Type & 
&\Prf : \Prop \ra \Type
\end{flalign*} 
We define an implication $\imp$, along with a rewrite rule that subsumes its natural deduction rules. 
\begin{flalign*}
&\imp : \Prop \ra \Prop \ra \Prop & \\
&\Prf ~(p \imp q) \lra \Prf ~p \ra \Prf ~q & 
\end{flalign*}
We also add the usual universal quantifier to exemplify the treatment of binders:
\begin{flalign*}
&\fa : (\tm \ra \Prop) \ra \Prop & \\
&\modif{\Prf ~(\fa ~p) \lra \Pi x : \tm. ~\Prf ~(p ~x)} & 
\end{flalign*}
\end{example}

Sorted logic arises by adding a constant $\Set$ for object-logic sorts and allows quantification over sorted object-logic terms.
There are two variants to define, going back to the definitions by Curry and Church, which we will refer to as soft-sorted logic \SFOL and hard-sorted logic \HFOL.

\begin{example}[Soft-Sorted Logic]
The basic structure of the theory \SFOL arises from that of \UFOL by adding a type $\Set$ of sorts and an external predicate $\#$ to capture the sorting of terms:
\begin{flalign*}
&\tm : \Type &
&\Set : \Type &
&\Prop : \Type &
&\Prf : \Prop \ra \Type & 
&\# : \tm \ra \Set \ra \Type 
\end{flalign*} 
The definition of the implication is the same as in \UFOL.
But we change the universal quantifier, which is unsorted in \UFOL, to a polymorphic version that takes as an additional argument the sort $a$ over which it quantifies.
\begin{flalign*}
&\fa : \Pi a : \Set. ~(\Pi x : \tm. ~\oft x a \ra \Prop) \ra \Prop & \\
&\modif{\Prf ~(\fa ~a ~p) \lra \Pi x : \tm. ~\Pi h : \oft x a. ~\Prf ~(p ~x ~h)} & 
\end{flalign*}
The body of the quantifier binds two assumptions: the bound variable and a proof that it has the sort $a$.
The latter ensures that bound variables are always well-sorted \modif{and thus that variables range over well-sorted objects only.
Similarly, any extension of \SFOL to $\lambda$-calculus requires guarding the bound variable of the $\lambda$-abstraction.}
\end{example}

\begin{example}[Hard-Sorted Logic]
Hard-sorted logic \HFOL uses the type of sorts $\Set$ and the injection $\El$ that maps a sort to the type of its elements.
\begin{flalign*}
&\Set : \Type &
&\El : \Set \ra \Type &
&\Prop : \Type &
&\Prf : \Prop \ra \Type 
\end{flalign*} 
Thus, object-logic terms of sort $a$ have type $\El ~a$.

The encoding of the implication is the same as for \UFOL and \SFOL.
The polymorphic universal quantifier of \HFOL differs from that of \SFOL by binding a sorted variable as a single variable of the framework:
\begin{flalign*}
&\fa : \Pi a : \Set. ~(\El ~a \ra \Prop) \ra \Prop & \\
&\alli : \Pi a : \Set. ~\Pi p : \El ~a \ra \Prop. ~(\Pi x : \El ~a. ~\Prf ~(p ~x)) \ra \Prf ~(\fa ~a ~p) & \\
&\alle : \Pi a : \Set. ~\Pi p : \El ~a \ra \Prop. ~\Prf ~(\fa ~a ~p) \ra \Pi x : \El ~a. ~\Prf ~(p ~x) & 
\end{flalign*}
\modif{Note that we encode the semantics of the universal quantifier here using two axioms as opposed to the rewrite rule used in \UFOL and \SFOL.
This is due to a technical issue that we will explain in \Cref{subsec_axvrew}.}
\end{example}

\subsection{From Hard-Sorted Logic to Soft-Sorted Logic}
\label{subsec_hard_soft}

\paragraph{\modif{Overview}}
\modif{To translate from \HFOL to \SFOL, the first intuition is to proceed in two steps: we erase the sorting information using a theory morphism, and then we recover it using a logical relation.
This approach is one of the example of~\cite{logical_relations}, with $\lambda$-binding instead of quantifiers.

The first step is to define a morphism, mapping sorts to themselves and \emph{erasing} the hard sort information by mapping every type $\El ~a$ to the type $\tm$. 
The constants $\Set$, $\Prop$ and $\Prf$ are mapped to themselves.
\begin{flalign*}
&\mor(\El) = \lambda a : \Set. ~\tm \\
&\mor(\fa) = \lambda a : \Set. ~\lambda p : \tm \ra \Prop. ~\fa ~a ~(\lambda x : \tm. ~\lambda h : \oft{x}{a}. ~p ~x) &
\end{flalign*} 
Then, in a second step, we recover the sort information by proving that whenever we have $t : \El ~a$ in \HFOL, we can show (i.e., give a term of type) $\oft{\mor(t)}{\mor(a)}$ in \SFOL. 
To do so, we define a unary logical relation on the morphism.
Terms of type $\Set$ are mapped by the morphism to themselves, so $\lr(\Set)$ can be any trivially satisfied predicate, for instance $\fa ~(\lambda p. ~p \imp p)$.
\begin{flalign*}
&\lr(\Set) = \lambda a : \Set. ~\Prf ~(\fa ~(\lambda p. ~p \imp p)) \\
&\mor(\El) = \lambda a : \Set. ~\lambda a^* : \Prf ~(\fa ~(\lambda p. ~p \imp p)). ~\lambda x : \tm. ~\oft{x}{a} &
\end{flalign*} 
So far, we have only translated the \emph{syntax} of the language.
To prove the translation sound, we must also translate the \emph{proof rules}.
However, we encountered a problem when extending the morphism to the proof rules.
The proof rule $\alle$ takes a term argument $x : \El ~a$.
Therefore $\mor(\alle)$ takes a term argument $x : \tm$.
But to define $\mor(\alle)$ in \SFOL, we must have access to the sort-preservation invariant $\oft{\mor(x)}{\mor(a)}$.
Thus, we cannot define the morphism for the proof rules without already using the logical relation on the syntax---whereas the logical relation must be defined after the theory morphism.
The translation in~\cite{logical_relations} worked out because it did not cover the proof rules.

One way out of this is to define a mutually recursive morphism and relation.
This is essentially the approach followed in~\cite{deduktinterp} for particular theories of \lpr.
We have generalized that idea to a systematic meta-theorem for \lpr.
It allowed the morphism to flexibly make use of the invariant established by the relation whenever it was needed.
But the use of this formalism became very complex.

Instead, we have identified a third approach as the most scalable: we pair up the morphism and the translation, and define both at once.
For example, we map $\El ~a$ to the dependent pair type $\Sigma x : \tm. ~\oft x a$.
Thus, every term in the image of the translation always carries its well-sortedness proof.
While seemingly heavyweight in the use of pairs, this approach has the key advantage that it can technically be represented as a theory morphism, thus keeping the framework simple.
}

Of course, the syntax of \lpr does not feature dependent pairs.
We could extend \lpr, but that would cut us off from implementations, like \dk, that do not have dependent pairs.
Alternatively, we could construct another kind of translation that eliminates dependent pairs, but that would again complicate the framework.
However, it turns out we only need very specific instances of dependent pairs for the translation that we can encode in \SFOL.

\paragraph{\modif{Extending the Target Language}}
For our specific translation we only need the types $\Sigma x : \tm. ~\oft x a$ for every $a : \Set$.
The declarations below extend \SFOL to $\SFOL'$ by adding this type, called $\pair ~a$.
\begin{flalign*}
&\pair : \Set \ra \Type & \\
&\mkpair : \Pi a : \Set. ~\Pi x : \tm. ~\oft x a \ra \pair ~a & \\
&\fst : \Pi a : \Set. ~\Pi m : \pair ~a. ~\tm & \\
&\snd : \Pi a : \Set. ~\Pi m : \pair ~a. ~\oft{(\fst ~a ~m)}{a} & 
\end{flalign*} 
\modif{
We also specify rewrite rules, effectively making $\pair ~a$ behave like $\Sigma x : \tm. ~\oft x a$.
}
\begin{flalign*}
&\fst ~a ~(\mkpair ~a ~x ~h) \lra x & \\
&\snd ~a ~(\mkpair ~a ~x ~h) \lra h & \\
&\mkpair ~a ~(\fst ~a ~m) ~(\snd ~a ~m) \lra m &
\end{flalign*} 
\modif{
In particular, the second rewrite rule only preserves typing because of the first rewrite rule: the left-hand side is of type $\oft{(\fst ~a ~(\mkpair ~a ~x ~h))}{a}$ while the right-hand side is of type $\oft{x}{a}$.

\begin{remark}
The combination of these rewrite rules with $\beta$-reduction is not confluent on untyped terms~\cite{klop_pairing}, but is confluent in several typed cases~\cite{pottinger,curien_cosmo}.
We conjecture that it is also confluent on typed terms of \lpr.
\end{remark}
}

Note that $\SFOL'$ is a conservative extension of \SFOL, in the sense that there is no $\SFOL$-type that is uninhabited over \SFOL but inhabited over $\SFOL'$.
Thus, $\SFOL'$ cannot prove any $\SFOL$-proposition that $\SFOL$ cannot prove.
In fact, we could even define $\pair$ if we worked in an \lpr-like framework with dependent pairs.

\paragraph{\modif{Defining the Translation}}
We can now give a morphism $\HS:\HFOL\to\SFOL'$.
\begin{flalign*}
&\mor(\Set) =  \Set & \\
&\mor(\El) =  \lambda a : \Set. ~\pair ~a & \\
&\mor(\Prop) = \Prop & \\
&\mor(\Prf) = \Prf &
\end{flalign*} 
Mapping the implication is straightforward, and the condition on the rewrite rule of $\imp$ is trivially satisfied.
To define $\mor(\fa)$, we have a predicate $p$ that takes a pair as an argument, but we need to use the universal quantifier of soft-sorted logic, in which the predicate takes two arguments sequentially.
The translations for $\alli$ and $\alle$ can be easily derived, as it suffices to pack elements into a pair or unpack them.
\[
\begin{array}{lll}
\mor(\fa) &= &\lambda a : \Set. ~\lambda p : \pair ~a \ra \Prop. ~\fa ~a ~(\lambda x. ~\lambda h. ~p ~(\mkpair ~a ~x ~h)) \\
\mor(\alli) &= &\lambda a : \Set. ~\lambda p : \pair ~a \ra \Prop. ~\\
  &&\lambda H : (\Pi m : \pair ~a. ~\Prf ~(p ~m)). \\
  &&\lambda x : \tm. ~\lambda h : \oft x a. \\
  &&H ~(\mkpair ~a ~x ~h) \\
\mor(\alle) &= &\lambda a : \Set. ~\lambda p : \pair ~a \ra \Prop. \\
  &&\lambda H : \Prf ~(\fa ~a ~(\lambda x. ~\lambda h. ~p ~(\mkpair ~a ~x ~h))) \\
  &&\lambda m : \pair ~a. \\
  &&H ~(\fst ~a ~m) ~(\snd ~a ~m)
\end{array}
\]
\modif{Note that we built a term of type $\Prf ~(p ~(\mkpair ~a ~(\fst ~a ~m) ~(\snd ~a ~m)))$, although $\mor(\alle)$ must return an object of type $\Prf ~(p ~m)$. 
These two types are convertible thanks to the third rewrite rule.}

\modif{Without rewriting, we would have to represent the rewrite rules of dependent pairs as equality axioms in the object-logic.
But these rewrite rules are used extensively in practice: because our translation pairs every term with the proof of its invariant, every operation on terms must unpair its arguments, operate on the components, and pair the results, thus introducing reducible expressions all over the place.
For example, \cite{IR:foundations:10} used a similar pairing approach to translate type theory into set theory.
But it used LF, and practical applications of that translation were bogged down by the resulting overhead of object-logic deduction steps to reduce these expressions.
With rewriting, we can minimize this overhead.
}

\subsection{From Soft-Sorted Logic to Unsorted Logic}
\label{subsec_soft_un}

\paragraph{\modif{Overview}}
To translate from \SFOL to \UFOL, the key intuition is to map every \emph{sort} to a unary \emph{predicate} on unsorted terms, and to use that predicate to relativize the quantifiers.
Then the external sorting relation $\oft t a$ can be mapped to the proposition $\mor(a) ~\mor(t)$. 

\modif{
The critical part of the translation is the treatment of bound variables in the universal quantifier, and this is an open problem for machine-verified language translations.
The intuitive idea is to translate a quantification $\fa ~a ~p$ over some sort $a$ with body $p$ to a relativized unsorted quantifier of the form $\fa ~(\lambda x. ~\mor(a) ~x \imp \mor(p) ~x)$.
But in \SFOL, the body $p$ takes two arguments---a term $x$ and a proof of $\oft x a$.
The translation of $p$ therefore takes two arguments---a term $x$ and a proof of $\mor(a) ~x$.
The straightforward relativization fails here: the application of $\mor(p)$ is ill-typed.

\begin{remark}
We have identified this as a very general issue.
Let us call an object-language expression \emph{proof-carrying} when it has a proof as a sub-expression.
Many languages use a stratified design where the syntax does not depend on the proof calculus and therefore proof-carrying expressions cannot arise.
If the source language allows proof-carrying expressions while the target language does not, it is impossible for the syntax translation to ever make use of any proofs carried by expressions.
\end{remark}

We could try to erase the assumption $\oft x a$, e.g., by translating it to a trivial type.
But then we would run into the same problems as discussed above for $\HFOL\to\SFOL$: we need these assumptions when translating the proof rules.
After much trial and error, we have identified one way to realize the translation at minimal cost: it again involves a subtle extension of the target language.
}

\paragraph{\modif{Extending the Target Language}}
We extend \UFOL to $\UFOL'$ by adding \emph{dependent} implication \cite{IR:foundations:10} where the construction of the second argument may assume the truth of the first.
Normally, dependent implication is useless in \UFOL because it is not possible to construct proof-carrying terms anyway.
But it makes a difference when translating \SFOL-terms that already carry sort assumptions.

$\UFOL'$ replaces the declarations regarding implication by an encoding of dependent implication~\cite{theoryU}.
\begin{flalign*}
&\impd : \Pi p : \Prop. ~(\Prf ~p \ra \Prop) \ra \Prop & \\
&\Prf ~(p \impd q) \lra \Pi h : \Prf ~p. ~\Prf ~(q ~h) & 
\end{flalign*}
\modif{Dependent implication is the Curry-Howard analogue of dependent functions, just like plain implication is the analogue of simple functions.}
It is easy to give a theory morphism from \UFOL to $\UFOL'$ that shows that the usual implication is a special case of the dependent one.

\paragraph{\modif{Defining the Translation}}
We now define a theory morphism $\SU:\SFOL\to\UFOL'$.
The translation of most constants is straightforward:
\begin{flalign*}
&\mor(\tm) = \tm & \\
&\mor(\Set) = \tm \ra \Prop & \\
&\mor(\Prop)  =\Prop & \\
&\mor(\Prf) = \Prf & \\
&\mor(\#) = \lambda x : \tm. ~\lambda a : \tm \ra \Prop. ~\Prf ~(a ~x) \\
&\mor(\imp) = \imp &
\end{flalign*} 
The condition on the rewrite rule of $\imp$ is trivially satisfied.

The translation of the universal quantifier is the key case.
It succeeds now because we can assume the relativizing predicate to hold when building the body:
\begin{flalign*}
&\mor(\fa) =\lambda a : \tm \ra \Prop. ~\lambda p : (\Pi x : \tm. ~\Prf ~(a ~x) \ra \Prop). ~\fa ~(\lambda x : \tm. ~(a ~x) \impd (\lambda h. ~p ~x ~h)) &
\end{flalign*} 

\modif{
The condition on the rewrite rule of $\fa$ is satisfied in $\UFOL'$.
\[
\begin{array}{lll}
\mor(\Prf ~(\fa ~a ~p)) &\equiv_{\beta\R} &\Prf ~(\fa ~(\lambda x : \tm. ~(a ~x) \impd (\lambda h. ~p ~x ~h))) \\
	&\equiv_{\beta\R} &\Pi x : \tm. ~\Prf ~((a ~x) \impd (\lambda h. ~p ~x ~h)) \\
	&\equiv_{\beta\R} &\Pi x : \tm. ~\Pi h : \Prf ~(a ~x). ~\Prf ~(p ~x ~h) \\
	&\equiv_{\beta\R} &\mor(\Pi x : \tm. ~\Pi h : \oft x a. ~\Prf ~(p ~x ~h))
\end{array}
\]
}

\begin{remark}
If we extend \SFOL with a function sort constructor $\arr : \Set \ra \Set \ra \Set$ and with term constructors for $\lambda$-abstraction and application, then the morphism can only be completed if \UFOL is also extended with appropriate unsorted $\lambda$-abstraction and application.
Critically, this unsorted $\lambda$-abstraction must have type $\lam : \Pi a : \tm \ra \Prop. ~(\Pi x : \tm. ~\Prf ~(a ~x) \ra \tm) \ra \tm$, where the first argument defines the domain of application, and where the bound variables are guaranteed to be from that domain.
The morphism fails with variants of \UFOL $\lambda$-abstraction that do not take into account the well-sortedness of $x$.
\end{remark}

\modif{
\subsection{Axioms vs. Rewrites}
\label{subsec_axvrew}

Along the lines of \Cref{ex_deduc_comput}, we can give both a computational and a deductive variant for \HFOL, \SFOL, and \UFOL.
And in each case, we can obtain a morphism $\mor_{\T}$ from the deductive to the computational variant of $\T$.

Typically, if the source theory is deductive, it does not matter if the target theory is deductive or computational---the translations can be established in essentially the same way.
In particular, if the target theory $\T$ is deductive, we can simply compose the translation with $\mor_{\T}$ to obtain a translation into the computational variant of $\T$.
But giving the translation into the computational variant directly can be much simpler because the rewriting in the target theory makes the necessary proofs a lot easier.

It would be unreasonable to expect a translation from a computational to a deductive variant, e.g., from computational \SFOL to deductive \UFOL: if the source theory already identifies certain types, they cannot be distinguished in the translation.

But we were surprised to find out that translation between computational variants can also pose subtle issues.
Indeed, we failed to give the translation from computational \HFOL to computational \SFOL, and that is why we used deductive \HFOL above.
To understand what goes wrong, consider the computational variant of the universal quantifier using the rewrite rule
\[\Prf ~(\fa ~a ~p) \lra \Pi x : \El ~a. ~\Prf ~(p ~x)\]
After applying \HS, the translations of the left-hand side and the right-hand side are not convertible in $\SFOL'$ and therefore the theory morphism is not well-formed: 
\[
\begin{array}{lll}
\mor(\Prf ~(\fa ~a ~p)) &\equiv_{\beta\R} &\Pi x : \tm. ~\Pi h : \oft{x}{\mor(a)}. ~\Prf ~(\mor(p) ~(\mkpair ~\mor(a) ~x ~h)) \\
\mor(\Pi x : \El ~a. ~\Prf ~(p ~x)) &\equiv_{\beta\R} &\Pi x : \pair ~a. ~\Prf ~(\mor(p) ~x)
\end{array}
\]
The problem here is, essentially, that the former is uncurried while the latter is curried.
To our knowledge, there is no rewrite-based system that would allow making these two types convertible.
Even in a hypothetical variant of \lpr with $\Sigma$-types, it is not obvious how such an identification up-to currying could be achieved efficiently.

It is of course possible to add $\Sigma$-types to \lpr, define the translation, and then systematically eliminate the $\Sigma$-types.
But the result can become very unwieldy because the elimination must split all functions that return $\Sigma$-types into two.
To fill this technical gap, we could consider an extension of \lpr with $\Sigma$-types available, but extend the conversion relation in such a way that they are rewritten into plain \lpr.
}

\section{Embedding Data Types}
\label{sec_nat_int}
\newcommand{\Nat}{\ident{Nat}}
\newcommand{\Int}{\ident{Int}}
\newcommand{\Den}{\ident{Den}}
\newcommand{\NI}{\ident{NI}}
\newcommand{\natden}{\ident{natden}}
\newcommand{\intden}{\ident{intden}}

\modif{Our use of dependent implication and dependent pairs in the previous translations is not a one-off trick.
They are a general techniques for realizing subtly difficult translations.
One common application is embedding theorems that identify a smaller data type as a fragment of a larger one.
Here the difficulty is that a binding over the smaller source data type must be relativized when translating it to a binding of the larger target data type.
As an example, we give an embedding translation from natural numbers to integers, in which we need to relativize bindings of integer variables with the property of being non-negative.

This property serves as an invariant the translated terms must satisfy.
Similar to the sort-erasure translations, it is critical to prove this invariant by induction over terms, and this is difficult because the translation must be intertwined with the proof of the invariant.
}

We start by encoding both theories as sorted languages, i.e., as extensions of \HFOL.

\begin{example}[Natural Numbers]
The theory $\HFOL+\Nat$ of natural numbers is built on hard-sorted logic.
But for the sake of example, we encode the proofs of universally quantified propositions computationally, i.e., via a rewrite rule:
\begin{flalign*}
&\fa : \Pi a : \Set. ~(\El ~a \ra \Prop) \ra \Prop & \\
&\Prf ~(\fa ~a ~p) \lra \Pi x : \El ~a. ~\Prf ~(p ~x) &
\end{flalign*}
We define the sort $\nat$ for natural numbers, with the two constructors $0$ and $\succ$. 
The relation $\geq$ is reflexive and transitive. 
\begin{flalign*}
&\nat : \Set & \\
&0 : \El ~\nat & \\
&\succ : \El ~\nat \ra \El ~\nat & \\
&\geq ~: \El ~\nat \ra \El ~\nat \ra \Prop & \\
&\ax_1 : \Pi x : \El ~\nat. ~\Prf ~(x \geq x) & \\
&\ax_2 : \Pi x,y,z : \El ~\nat. ~\Prf ~(x \geq y) \ra \Prf ~(y \geq z) \ra \Prf ~(x \geq z) & 
\end{flalign*}
For any natural number $x$, $\succ ~x$ is greater than $x$. In other words, we have a proof of $\succ ~x \geq x$, and any proof of $x \geq \succ ~x$ leads to an inconsistency.
\begin{flalign*}
&\ax_3 : \Pi x : \El ~\nat. ~\Prf ~(\succ ~x \geq x) & \\
&\ax_4 : \Pi x : \El ~\nat. ~\Prf ~(x \geq \succ ~x) \ra \Pi P : \Prop. ~\Prf ~P &
\end{flalign*}
Finally, we have the induction principle on natural numbers:
\[
\begin{array}{ll}
\rec : &\Pi P : \El ~\nat \ra \Prop. ~\Prf ~(P ~0) \ra \\
 &[\Pi x : \El ~\nat. ~\Prf ~(P ~x) \ra \Prf ~(P ~(\succ ~x))] \ra \\
 &\Pi x : \El ~\nat. ~\Prf ~(P ~x) 
\end{array}
\]
\end{example}

\begin{example}[Integers]
The theory of integers $\HFOL+\Int$ is like $\HFOL+\Nat$, with the sort $\nat$ renamed to $\int$. Additionally, we introduce a predecessor symbol $\pred$, such that $\pred$ and $\succ$ are inverses.
\begin{flalign*}
&\pred : \El ~\int \ra \El ~\int & \\
&\succ ~(\pred ~x) \lra x & \\
&\pred ~(\succ ~x) \lra x &  
\end{flalign*}
\modif{The normal terms of type $\El ~\int$ are in bijection to the integers.}

For any integer $x$, $\pred ~x$ is lower than $x$. 
\begin{flalign*}
&\ax_5 : \Pi x : \El ~\int. ~\Prf ~(x \geq \pred ~x) & \\
&\ax_6 : \Pi x : \El ~\int. ~\Prf ~(\pred ~x \geq x) \ra \Pi P : \Prop. ~\Prf ~P &
\end{flalign*}
The induction principle on integers
\[
\begin{array}{ll}
\rec : &\Pi P : \El ~\int \ra \Prop. ~\Prf ~(P ~0) \ra \\
 &[\Pi x : \El ~\int. ~\Prf ~(x \geq 0) \ra \Prf ~(P ~x) \ra \Prf ~(P ~(\succ ~x))] \ra \\
 &[\Pi x : \El ~\int. ~\Prf ~(0 \geq x) \ra \Prf ~(P ~x) \ra \Prf ~(P ~(\pred ~x))] \ra \\
 &\Pi x : \El ~\int. ~\Prf ~(P ~x) 
\end{array}
\]
involves an additional induction step for $\pred$.
\end{example}

The translation $\HFOL+\Nat\to\HFOL+\Int$ is structurally very similar to the one $\HFOL\to\SFOL$. 
We intuitively map the sort $\nat$ to the sort $\int$, but we need to recover the information that any natural number is mapped to a non-negative integer. 
Here the invariant of the translation is the predicate $x \geq 0$.
Like the translation $\HFOL\to\SFOL$, we will employ dependent pairs.
The only dependent pairs we need are of the form $\Sigma x : \El ~a. ~\Prf ~(p ~x)$.
The definitions below conservatively extend $\HFOL+\Int$ to $\HFOL+\Int'$ with an axiomatization of those dependent pairs.
\begin{flalign*}
&\pair : \Pi a : \Set. ~(\El ~a \ra \Prop) \ra \Set & \\
&\mkpair : \Pi a : \Set. ~\Pi p : \El ~a \ra \Prop. ~\Pi x : \El ~a. ~\Prf ~(p ~x) \ra \El ~(\pair ~a ~p) & \\
&\fst : \Pi a : \Set. ~\Pi p : \El ~a \ra \Prop. ~\El ~(\pair ~a ~p) \ra \El ~a & \\
&\snd : \Pi a : \Set. ~\Pi p : \El ~a \ra \Prop. ~\Pi m : \El ~(\pair ~a ~p). ~\Prf ~(p ~(\fst ~a ~p ~m)) & \\
&\fst ~a ~p ~(\mkpair ~a ~p ~x ~h) \lra x & \\
&\snd ~a ~p ~(\mkpair ~a ~p ~x ~h) \lra h & \\
&\mkpair ~a ~p ~(\fst ~a ~p ~m) ~(\snd ~a ~p ~m) \lra m &
\end{flalign*}
Then we can define $\NI:\HFOL+\Nat\to\HFOL+\Int'$. The constants of hard-sorted logic are mapped to themselves.
The sort of natural numbers is mapped to the sort that pairs an integer and a proof that it is non-negative.
\begin{flalign*}
&\mor(\nat) = \pair ~\int ~(\lambda x. ~x \geq 0) & \\
&\mor(0) = \mkpair ~\int ~(\lambda x. ~x \geq 0) ~0 ~(\ax_1 ~0) & \\
&\mor(\geq) = \lambda m_1,m_2 : \pair ~\int ~(\lambda x. ~x \geq 0). ~(\fst ~\int ~(\lambda x. ~x \geq 0) ~m_1)\geq (\fst ~\int ~(\lambda x. ~x \geq 0) ~m_2) & \\
&\mor(\ax_1) = \lambda m : \pair ~\int ~(\lambda x. ~x \geq 0). ~\ax_1 ~(\fst ~\int ~(\lambda x. ~x \geq 0) ~m) & 
\end{flalign*}
Most of the remaining parameters are defined similarly.
The parameter $\mor(\rec)$ is trickier: it requires a dependent implication, for the same reason as the translation $\HFOL\to\SFOL$. 
It also requires proof irrelevance---the principle stating that two proofs of the same proposition are equal. 
We add dependent implication and an axiom for proof irrelevance to $\HFOL+\Int'$.
\begin{flalign*}
&\impd : \Pi p : \Prop. ~(\Prf ~p \ra \Prop) \ra \Prop & \\
&\Prf ~(p \impd q) \lra \Pi h : \Prf ~p. ~\Prf ~(q ~h) & \\
&\mathsf{proof\_irr} : \Pi p : \Prop. ~\Pi h_1,h_2 : \Prf ~p. ~\Pi q : \Prf ~p \ra \Prop. ~\Prf ~(q ~h_1) \ra \Prf ~(q ~h_2) &
\end{flalign*}
The translation from natural numbers to integers was the running example of~\cite{deduktinterp}, where it was formalized using an intricate construction involving mutually recursive definitions of morphism and logical relation.
It also required more boilerplate such as trivially true invariants that must be carried through the induction.
In contrast, the present translation is itself much simpler  and can be expressed in a simpler framework.

\modif{Our embedding shows that natural numbers are mapped to non-negative integers.
We can extend the theories \Nat and \Int with operations such as addition, and extend the morphism accordingly.
Then the morphism additionally shows that the corresponding operations on integers preserve the invariant.
}

\modif{Moreover, we can use a logical relation to show that the denotations of natural numbers are preserved by the embedding.
First, assume we have a theory $\Den$ in which denotations of numbers are defined.
The semantics of numbers can be represented by morphisms $\natden:\Nat\to\Den$ and $\intden:\Int\to\Den$.
To show the preservation of denotations, we can then give a binary logical relation on the morphisms $\natden$ and $\NI;\intden$, that maps the type $\nat$ to the statement of equality of the two denotations.
}



\section{Implementation for Dedukti}
\label{sec_implem}
\paragraph{Implementation}
\dk\footnote{Available at \url{https://github.com/Deducteam/Dedukti}.} is a proof language based on \lpr. 
We developed a tool, called \texttt{TranslationTemplates}\footnote{Available at \url{https://github.com/Deducteam/TranslationTemplates}.}, that implements the new features developed here.

Because the main applications of \dk are the batch processing of large sets of theorems, our design makes the same trade-offs and optimizes for the batch transport of theorems from a source theory $\S$ to a source theory $\T$.
\texttt{TranslationTemplates} takes two files representing $\S$ and $\T$, and outputs a new file that contains a copy of $\S$ as an extension of $\T$.
All primitive declarations of $\S$ result in gaps that the user needs to fill in---these are the parameters of a theory morphism/logical relation.
Then all defined declarations of $\S$ can simply be copied over.
An additional argument controls if the generated file should capture a theory morphism or a logical relation.
The resulting file can be rechecked by \dk so that our code does not have to be a part of the trusted code base.
The tool is written in OCaml in less than $400$ lines of code and benefits from the \dk kernel and parser.
All the examples of theory morphisms given here have been implemented in \dk and mechanically checked.
\modif{As an additional example, we have implemented a theory morphism from higher-order classical logic to higher-order intuitionistic logic, following~\cite{kurodadukti}.}

\paragraph{Demonstration} We illustrate \texttt{TranslationTemplates} on the theory morphism from \Cref{ex_deduc_comput}. 
For simplicity, we only consider the implication symbol. The source file \verb|deduction.dk| contains the theory where natural deduction rules are encoded via axioms.
\begin{lstlisting}
Prop : Type.
Prf : Prop -> Type.
imp : Prop -> Prop -> Prop.
imp_i : p : Prop -> q : Prop -> (Prf p -> Prf q) -> Prf (imp p q).
imp_e : p : Prop -> q : Prop -> Prf (imp p q) -> Prf p -> Prf q.
\end{lstlisting}
\modif{We prove the theorem $\Pi p : \Prop. ~\Prf ~(p \imp p)$ by taking the proof term $\lambda p : \Prop. ~\mathsf{imp_i} ~p ~p ~(\lambda H : \Prf ~p. ~H)$.}
\begin{lstlisting}
thm lemma_imp : p : Prop -> Prf (imp p p)
  := p => imp_i p p (H => H).
\end{lstlisting}
The target file \verb|computation.dk| contains the theory where natural deduction rules are encoded via rewrite rules. 
Note that the symbol \verb|Prf| is now declared with \verb|def|, because it is definable with rewrite rules.
\begin{lstlisting}
Prop : Type.
def Prf : Prop -> Type.
imp : Prop -> Prop -> Prop.
[p, q] Prf (imp p q) --> Prf p -> Prf q.
\end{lstlisting}
The theory morphism from \verb|deduction.dk| to \verb|computation.dk| generates the following file. 
\begin{lstlisting}
#REQUIRE computation.

def Prop_mu : Type := TODO.
def Prf_mu : Prop_mu -> Type := TODO.
def imp_mu : Prop_mu -> Prop_mu -> Prop_mu := TODO.
def imp_i_mu : p : Prop_mu -> q : Prop_mu -> 
               (Prf_mu p -> Prf_mu q) -> Prf_mu (imp_mu p q)
  := TODO.
def imp_e_mu : p : Prop_mu -> q : Prop_mu -> 
               Prf_mu (imp_mu p q) -> Prf_mu p -> Prf_mu q
  := TODO.
  
thm lemma_imp_mu : p : Prop_mu -> Prf_mu (imp_mu p p)
  := p => imp_i_mu p p (H => H).
\end{lstlisting}
This is a skeleton of the translation, where the parameters must be filled in by the user instead of the \verb|TODO|s. 
\modif{To help the user, the type of each parameter is displayed.}
Such parameters must be expressed in the target theory \verb|computation.dk|.
We can do so following the blueprint of \Cref{ex_deduc_comput}.
The statement and the proof of the theorem \verb|lemma_imp| of the source theory have been \textit{automatically} translated to the target theory. 
The file typechecks provided that the conditions of the theory morphism are fulfilled. 


\section{Conclusion}\label{sec:conc}

\paragraph{Summary}
We have introduced two translation templates---based on theory morphisms and logical relations---for representing meta-theorems for the \lpcm.
Barring an example of a theory morphism in \cite{felicissimo_encodings}, this is the first time that such translation templates are used systematically for a logical framework with rewriting.

\modif{
Throughout our examples, we have dealt with the subtle trade-off between axioms and rewrite rules.
On the one hand, when using rewrite rules instead of axioms, parts of the proofs are directly discharged by the framework through computation.
In particular, the presence of rewrite rules in the target theory simplifies the proof obligations generated by the translation templates.
On the other hand, the presence of rewrite rules in the source theory creates additional constraints that need to be satisfied by the translation templates.
These observations allow us to fully leverage the computational power of the \lpcm to formalize translations between theories.
}

Moreover, we have identified two subtle practices that allow representing meta-theorems that have previously proved challenging:
the use of dependent pairs (as a primitive of the framework or as an ad-hoc conservative extension) and of dependent implication.
This observation is independent of rewriting and applies to other logical frameworks as well.
More generally, it indicates that efficient translations across languages may be critically enabled by deep technical tweaks to the framework or the target language.

We have implemented both templates in the \dk proof language, and we have applied them to mechanically check translations between a number of logics.
These kinds of templates are critical for the interoperability of proof systems, a major goal of the \dk project, as they allow for the batch translation of large libraries along a theory morphism or a logical relation.

\modif{
\paragraph{Reflection Properties}
Both morphisms and logical relations capture preservation-style meta-theorems (``if it holds in the source, it is preserved in the target'').
These correspond to soundness proofs if we think of the source as the syntax and the morphism as the interpretation function that maps the syntax to its semantics.
Such proofs typically proceed by induction on derivations, and that recipe is exactly what morphisms and relations capture.

The dual, reflection-style meta-theorems (``if it holds in the target, it already held in the source'') are typically much harder to prove and would therefore be more interesting to verify in logical frameworks.
Such theorems correspond to completeness theorems.
They are called conservativity of morphisms in~\cite{rabe:cons:24}.

Morphisms can only express that the target theory is strong enough to express the source theory.
Sometimes a pair of isomorphisms can be used to express preservation and reflection and thus show that two theories can define each other.
But in most cases a morphism showing preservation is not an isomorphism even if some reflection property holds.

This is because the reflection property often holds only in some restricted form and not for all expressions.
For example, we can give a morphism $\mor$ from \HFOL to set theory that captures the model-theoretical semantics of the logic as well as its soundness proof.
The completeness of the semantics is a restricted reflection property: $\mor$ reflects provability in the sense that if $\mor(\Prf~p)$ is inhabited then so is $\Prf~p$.
But there is no inverse morphism $\mor^{-1}$ that maps all of set theory back to \HFOL.


One way to formalize such restricted reflection theorems is to give a partial translation from the target theory to the source theory.
For example, \cite{lfcut} gives and verifies such a translation to reflect derivability from a calculus with cut to one without.
If both the source and the target theory have a sound and complete semantics, an alternative is to relate those semantics to each other, as illustrated in~\cite{logical_relations}.
For both kinds of arguments, it is difficult to provide general formalization support at the logical framework level.
We see this as a major challenge problem for logical frameworks.
}

\paragraph{Acknowledgments.}
This publication is based upon work from the action CA20111 \href{https://europroofnet.github.io/}{EuroProofNet} supported by \href{https://www.cost.eu/}{COST} (European Cooperation in Science and Technology).

We thank the reviewers for their relevant comments and suggestions.

\bibliographystyle{alpha}
\bibliography{biblio}

\end{document}